\documentclass[11pt]{article}
\usepackage[margin=1in]{geometry}
\usepackage{amsmath, amssymb, amsthm}
\usepackage{graphicx}
\graphicspath{{figures/}{./}}
\usepackage{booktabs}
\usepackage{float}
\usepackage{algorithm}
\usepackage{caption}
\usepackage[round]{natbib}
\usepackage{setspace}
\usepackage{microtype}
\newtheorem{theorem}{Theorem}
\newtheorem{proposition}{Proposition}
\newtheorem{corollary}{Corollary}
\newtheorem{lemma}{Lemma}
\newtheorem{assumption}{Assumption}
\newtheorem{definition}{Definition}
\theoremstyle{remark}
\newtheorem{remark}{Remark}

\title{\textbf{Trading Scope for Credibility in Difference-in-Differences}}
\author{Parush Arora\thanks{Department of Economics, Ashoka University.} \and Abhishek Chand\thanks{Independent researcher.}}
\date{}

\begin{document}
\maketitle

\begin{abstract}
\noindent In staggered difference-in-differences, parallel trends may hold for some treated cohorts and
fail for others. The standard estimator is then biased for the average treatment effect on the treated
(ATT), which averages over all cohorts including the offenders. Honest inference that keeps the ATT as the
target does not repair this: it reads a pooled pre-trend that the cohort-level violations can wash out of,
leaving its confidence set either wide or, when the maintained sensitivity class is calibrated on that
washed-out aggregate, misleadingly narrow. We take a different route. We drop the offending cohorts and
target the effect for the credible ones, the difference-in-differences analogue of the local average
treatment effect and of overlap weighting. This credible-subpopulation effect is point-identified under a
weaker parallel-trends requirement that can hold when the ATT's fails. When the credible cohorts are read
off the observed pre-trends, however, selection reintroduces the pre-test problem the honest-DiD literature
was built to avoid, and existing honest methods do not cover it because they never select. Our
contributions are to formalize this change of estimand for staggered designs and to supply honest inference
for it. We compose the post-selection carving of \citet{lee2016} with the sensitivity bounds of
\citet{rambachan2023}: the carving makes inference for the selected reduced-form aggregate uniformly valid
against the data-driven selection, with no separation or consistent-selection condition when the
event-study covariance is known, and the sensitivity bounds carry it to the causal target at each level of
the residual violation the screen cannot rule out. We characterize in closed form when narrowing the target
lowers risk, a feasible rule that can cost efficiency but never coverage. In an application to the shale
boom, a relative-magnitudes analysis certifies a significantly positive pooled effect on house prices, but
the credible subpopulation, the counties whose oil-and-gas onset did not coincide with the mid-2000s
housing run-up, shows no effect, locating the pooled positive in cohorts already trending before onset.
\end{abstract}

\medskip
\noindent\textbf{Keywords. } Difference-in-differences, Staggered adoption, Parallel trends, Robust inference, Sensitivity analysis, Partial identification.

\smallskip
\noindent\textbf{JEL classification. } C1, C23, C51.

\newpage

\section{Introduction}

Difference-in-differences is among the most widely used research designs in empirical economics, and its
credibility rests on the parallel trends assumption, that absent treatment treated and comparison units
would have evolved in parallel. Researchers assess it by testing for pre-treatment differences in trends,
but that practice is fragile. Pre-trends tests are often underpowered against economically meaningful
violations \citep{bilinski2020, freyaldenhoven2019, kahnlang2020, roth2022pretrends}. Conditioning on
passing them induces a pre-test selection bias \citep{roth2022pretrends}, though such conditional inference
remains valid, if conservative, when parallel trends holds \citep{dechaisemartin2024pretest}. And even a clean pre-period
does not guarantee a clean post-period, since nothing prevents a confound from arriving at the moment of
treatment \citep{kahnlang2020}.

A growing literature, surveyed in \citet{roth2023trending}, responds by relaxing parallel trends while
holding the ATT fixed as the target. Building on \citet{manski2018}, \citet{rambachan2023} bound how far
post-treatment violations can depart from the observed pre-trends and report a uniformly valid confidence
set for the ATT. \citet{kwonroth2024} sharpen the set with an empirical-Bayes prior on the violation, and
\citet{dechaisemartin2026} rank the post-treatment differences against the pre-treatment ones. All keep
the ATT and use the pre-trends to extrapolate the violation across periods. But the ATT is where a
structural difficulty lies. When parallel trends fails for some cohorts and holds for others, the ATT is
an average over all treated cohorts, including the offenders, and its point estimate is biased. Staggered
adoption makes such heterogeneity the rule rather than the exception. Each cohort's timing must be defended
on its own, one cohort's adoption perhaps as good as randomly timed while another's coincides with a
confounding shock, so parallel trends is naturally a cohort-level property, credible for some and not for
others. Honest inference that keeps the ATT as the target does not repair this. It reads a pre-trend that pools the
cohorts, and the failure takes one of two forms. When the aggregate pre-trend reveals the violation, the
honest set is valid but wide, accommodating the worst cohort. When the offending cohorts' pre-trends wash
out of the aggregate, a minority, differently timed, the honest set is instead falsely precise, tight
around the biased estimate. As Section~\ref{sec:app} shows, a relative-magnitudes sensitivity analysis then
certifies a spurious effect. Either way the fixed-target analysis on the aggregate is the wrong instrument,
because the violation is a property of the cohort composition that the aggregate hides.

Our response is to change the target. When parallel trends fails for some cohorts, we drop them and report
the effect for the credible ones. The obstacle is inference. Selecting the credible cohorts from the
pre-trends revives the very pre-test bias honest-DiD exists to avoid, and no existing honest method
addresses it, because none selects. This paper supplies that inference. Applied practice already retreats
to the part of a design it trusts, but piecemeal and without inference to match. Researchers exclude
treated units whose inclusion would contaminate the comparison, select and weight comparison units for
pre-treatment fit \citep{abadie2010}, and, less formally, adjust samples and specifications until the
pre-trends look flat, the specification search whose statistical costs \citet{roth2022pretrends}
documents. Each is a retreat to a more credible subpopulation, made ad hoc or on design grounds, and none
carries inference valid against the selection it performs. When the population ATT is not credibly
identified, we make that retreat explicit and give it a target and an inference. We drop to the
subpopulation of cohorts where identification is credible, the logic of the local average treatment
effect \citep{imbens1994}, of the optimal-subpopulation average treatment effect under limited overlap
\citep{crump2009}, and of overlap weighting \citep{li2018}, each of which trades the original estimand for
a credibly identified effect on a selected subpopulation. The credible-subpopulation LATT is a reweighting
of standard group-time effects \citep{callaway2021} over the cohorts whose parallel trends is
credible.\footnote{These estimators
\citep{callaway2021, sun2021, borusyak2024, dechaisemartin2020} repair the aggregation failure of
two-way fixed effects under staggered adoption \citep{goodmanbacon2021}, whose implicit weights can turn
negative under treatment-effect heterogeneity. That repair concerns aggregation, not identification,
since they maintain parallel trends and are biased when it fails, which is the failure we address.}
Credibility can be decided ex ante, from covariates or institutional knowledge, but in the leading case it
is read from the pre-trends themselves, the same pre-trends used for estimation, which is exactly the
selection the honest-DiD literature avoided. The paper thus makes two contributions. The first is to
propose and formalize this change of estimand for staggered difference-in-differences, giving the
credible-subpopulation effect an identification result and a rule for when to adopt it. The second, where
the work lies, is to make the data-driven selection honest, with inference for the LATT that is uniformly
valid against both the selection and the residual violation of parallel trends that a pre-trend screen
cannot rule out.

That the LATT is point-identified under parallel trends on the selected cohorts alone, a weaker
requirement that can hold when the ATT's fails, is immediate. Because the
screen selects on pre-trends, and a flat pre-trend does not certify the flat post-trend that identification
requires, the honest interval, not the point, is the object we carry. This second contribution, the
inference a data-driven selection demands, has three parts. First, and centrally, honest post-selection
inference. We combine the post-selection
``carving'' of \citet{lee2016} with the sensitivity bounds of \citet{rambachan2023} into a single interval
for the LATT that is uniformly valid in the data-generating parameter, with no separation or
consistent-selection condition when the event-study covariance is known, and to first order when it is
consistently estimated. The carved component absorbs the data-driven
selection and the level bound absorbs the residual identification failure, and the two compose. Neither
ingredient suffices alone. Post-selection carving \citep{lee2016} delivers valid inference when the target
is point-identified conditional on the selected set. Here it is not, because a flat pre-trend does not
certify the flat post-trend that identification requires, so carving alone leaves a residual violation it
cannot touch, and composing it with a sensitivity class is what restores honesty. This is the step the
fixed-target sensitivity analysis never takes, because it never selects, and it carries selection and
identification uncertainty on separate instruments. Second, a feasible
decision rule.
We characterize in closed form when the trade pays, when the credible-subpopulation honest interval is
shorter than the honest interval for the ATT. A width-dominance theorem gives the explicit threshold on the dropped cohorts' violation
above which the credible-subpopulation interval is strictly shorter than the honest ATT interval, a
threshold estimable from the covariance and the chosen sensitivity bounds. The closed-form threshold is for
the fixed-length intervals of ex-ante selection. Under data-driven selection the carved interval's expected
length is infinite, so the comparison there is by coverage and median length rather than a
width threshold. We turn the characterization into an estimable rule reported alongside both
honest intervals. Because both intervals are honest, the rule can never cost
coverage, only efficiency. Under homogeneous effects it selects the lower-risk procedure with probability
approaching one, and under heterogeneity it returns a breakdown value $\Gamma^\ast$ in the idiom of a
sensitivity analysis. Third, evidence. Simulations trace the method's advantage across the
informativeness of pre-trends and confirm the width dominance directly. Once the dropped violation is large
enough, the honest credible-subpopulation interval is strictly shorter than the honest ATT interval, both
covering. An
application to the shale boom then withdraws a spuriously positive pooled estimate of the effect on house prices.
A relative-magnitudes sensitivity analysis certifies a $+6.7$ log-point effect, yet the credible
subpopulation, the counties whose onset did not coincide with the housing run-up, reveals essentially zero,
locating the pooled effect in the cohorts already trending before onset. An appendix develops the optimal credible
weighting that the same sensitivity class implies, a soft-thresholded
weighting of cohorts by credibility net of noise, of which the flatness screen is a special case,
completing the analogy to the variance-minimizing overlap weights of \citet{crump2009}.

Among existing robust-DiD methods, which hold the ATT fixed and use pre-trends to bound or extrapolate the
post-treatment violation, none selects, so none confronts selection in its inference. We instead change the
target and treat the resulting selection as the object requiring inference, producing an interval uniformly
valid against both the data-driven selection and the residual violation, a guarantee the fixed-target
methods have no need for. We nest the sensitivity analysis of
\citet{rambachan2023} by applying it to the retained sub-aggregate, and we connect to the post-selection
inference literature \citep{lee2016, leeb2005, leeb2008} that the honest-DiD papers, by never selecting,
did not need. The most closely related work is \citet{dechaisemartin2026}, which also brings the pre-trends into
the inference rather than leaving them to a pre-test, but keeps the ATT and bounds it by ranking the
post-treatment differences against the pre-treatment ones. We instead change the target to the credible
subpopulation, so the object inference must be valid against is the data-driven cohort selection, not the
extrapolation across periods.

The remainder of the paper is organized as follows. Section~\ref{sec:approach} fixes the staggered-DiD
setup, defines the credible-subpopulation LATT, and develops its identification and estimation.
Section~\ref{sec:honest} develops the honest inference that a data-driven selection demands, and
Section~\ref{sec:dominance} the decision of when to narrow the target. Section~\ref{sec:sim} reports the
simulation study,
Section~\ref{sec:app} applies the method to the local house-price effect of the shale boom, and
Section~\ref{sec:conc} concludes.

\section{The credible-subpopulation LATT}\label{sec:approach}

\subsection{Setup}\label{sec:setup}

We adopt the staggered adoption framework of \citet{callaway2021}. There are periods
$t=1,\dots,T$ and units $i$ indexed by their first treatment date $G_i \in \{2,\dots,T\}\cup\{\infty\}$,
where $G_i=\infty$ denotes never-treated. Let $Y_{it}(g)$ denote the potential outcome in period $t$
if unit $i$ is first treated in period $g$, and $Y_{it}(\infty)$ the never-treated potential outcome.
The building-block causal parameter is the group-time average treatment effect on the treated,
\begin{equation}
\mathrm{ATT}(g,t) = \mathbb{E}\!\left[\,Y_{it}(g)-Y_{it}(\infty)\mid G_i=g\,\right].
\end{equation}
Summary parameters are weighted aggregations $\theta = \sum_{g}\sum_t w(g,t)\,\mathrm{ATT}(g,t)$ with
researcher-chosen weights \citep{callaway2021}, of which the overall ATT and the event-study path are
special cases.

It is convenient to work with the event-study representation and the decomposition of
\citet{rambachan2023}. Let $\hat\beta=(\hat\beta_{\mathrm{pre}}',\hat\beta_{\mathrm{post}}')'$ be an
asymptotically normal vector of event-study coefficients estimated by any of the heterogeneity-robust
procedures above, with $\sqrt{n}(\hat\beta-\beta)\to\mathcal{N}(0,\Sigma)$. The finite-sample covariance of
$\hat\beta$ is then $\Sigma_n=\Sigma/n$, and the Gaussian working model of Section~\ref{sec:honest},
together with the standard errors, truncation limits, and randomization scale used throughout, is stated in
terms of $\Sigma_n$, and we write $\Sigma$ for $\Sigma_n$ when no confusion arises. The estimand decomposes as
\begin{equation}
\beta = \tau + \delta, \qquad \tau_{\mathrm{pre}}=0,
\label{eq:decomp}
\end{equation}
where $\tau$ collects the causal effects (zero before treatment, by no anticipation) and $\delta$ is
the differential trend between treated and comparison groups that would have occurred absent
treatment. Parallel trends is the restriction $\delta_{\mathrm{post}}=0$, under which
$\beta_{\mathrm{post}}=\tau_{\mathrm{post}}$, and a pre-trends test is a test of $\delta_{\mathrm{pre}}=0$.

\begin{assumption}[Maintained conditions]\label{ass:main}
(i) Asymptotic normality. The event-study estimator satisfies
$\sqrt{n}(\hat\beta-\beta)\xrightarrow{d}\mathcal N(0,\Sigma)$ for a positive-definite $\Sigma$, consistently
estimable by some $\hat\Sigma\xrightarrow{p}\Sigma$. (ii) No anticipation. $\tau_{g,\mathrm{pre}}=0$
for every cohort $g$, so that the pre-treatment coefficients identify the pre-treatment differential trend,
$\beta_{g,\mathrm{pre}}=\delta_{g,\mathrm{pre}}$.
\end{assumption}

All results maintain Assumption~\ref{ass:main}. The finite-sample carved inference of
Section~\ref{sec:honest} additionally works in the Gaussian model $\hat\beta\sim\mathcal N(\beta,\Sigma)$
with $\Sigma$ known, with the estimated-$\Sigma$ case treated in Remark~\ref{rem:sigmahat}.

We will require the analogous objects at the cohort level. Write $\delta_g$ for the differential
trend of cohort $g$ and say that parallel trends holds for $g$ if $\delta_{g,\mathrm{post}}=0$.
A key feature of applications, motivated in the introduction, is that this may hold for some cohorts and
fail for others, since each cohort's adoption timing must be defended on its own.

Finally, we distinguish two regimes for how cohorts are judged credible, because they carry different
inferential guarantees. Under ex-ante (or independent) selection, credibility is decided from
pre-determined covariates, institutional knowledge, or an independent data split, so that the selected
set is a fixed $S$ statistically independent of the estimation errors in $\hat\beta$. Under data-driven
selection, credibility is assessed from the realized pre-trends $\hat\beta_{\mathrm{pre}}$, so that the
selected set $\hat S$ is random and depends on the estimation sample, with a deterministic population
counterpart $S^\ast$, the set a population pre-trend statistic would select. The identification argument
below holds for any set held fixed, and so identifies the random target $\theta_{\hat S}$ conditional on the
realized $\hat S$. Identifying the fixed population parameter $\theta_{S^\ast}$ instead requires $\hat S$ to
recover $S^\ast$, a selection-consistency condition we return to under inference
(Remark~\ref{rem:oracle}). The distinction is thus minor for identification but central for inference,
and we treat it explicitly below.

\subsection{Estimand and identification}

Let $S$ denote a set of cohorts judged to have credible parallel trends, produced by a selection rule
$R$ that we specify below. We define the target as the treatment effect for that subpopulation.

\begin{definition}[Credible-subpopulation LATT]
For cohort weights $w_g>0$ (e.g.\ proportional to cohort size) and a within-cohort aggregation $a(\cdot)$, a
linear functional of a cohort's post-treatment vector (a fixed event-time, selecting one coordinate, or an
average over post-treatment periods), the credible-subpopulation local ATT, defined for any nonempty $S$, is
\begin{equation}
\theta_S \;=\; \frac{\sum_{g\in S} w_g\,\mathrm{ATT}(g,\cdot)}{\sum_{g\in S} w_g},
\qquad \mathrm{ATT}(g,\cdot)\;=\;a\!\left(\tau_{g,\mathrm{post}}\right),
\end{equation}
the aggregated causal effect of cohort $g$.
\end{definition}

The target is identified under a weaker condition than the ATT, a benchmark we state before turning,
in Section~\ref{sec:honest}, to the honest interval that is the paper's actual object.

\begin{proposition}[Identification benchmark]\label{prop:id}
If parallel trends holds for every $g\in S$, i.e.\ $\delta_{g,\mathrm{post}}=0$ for all $g\in S$, then
$\theta_S$ is point-identified by the reweighted aggregated group-time effects,
$\theta_S = \big(\sum_{g\in S} w_g\,a(\beta_{g,\mathrm{post}})\big)/\sum_{g\in S} w_g$, and this holds
irrespective of whether $\delta_{g,\mathrm{post}}\neq 0$ for cohorts $g\notin S$.
\end{proposition}

\begin{proof}
See Appendix~\ref{app:proofs}.
\end{proof}

The identification in Proposition~\ref{prop:id} is immediate. The substance is the honest inference for the
data-driven selection this target invites, developed in Section~\ref{sec:honest}. Changing the target to the
credible subpopulation weakens the identifying requirement from parallel trends on all cohorts to parallel
trends on $S$, which can hold when the former fails, at the cost of a narrower question whose subpopulation
$S$ must then be characterized. The weakening is genuine but conditional, and two things must be kept
separate. The identifying content is the substantive restriction $\delta_{g,\mathrm{post}}=0$ for $g\in S$,
an assumption about post-treatment trends. The flatness screen of Section~\ref{sec:select} is a statistical
classification rule that supplies evidence for that restriction from the pre-period, not the restriction
itself, since a cohort flat before treatment may still drift after it. Point identification of the causal
$\theta_S$ thus holds only when the restriction does, and the screen cannot certify it. We therefore treat
the point-identified $\theta_S$ as a benchmark and pair it with the honest sensitivity bounds developed
below, which bound the residual post-treatment violation the screen leaves behind and, in doing so,
return a set rather than a point when that violation is nonzero.

Narrowing the target carries an interpretability cost, and a standing reporting rule addresses it. The
credible subpopulation is delimited by a statistical property, flat pre-trends, rather than by a
substantive behavioral criterion, so it is less transparently interpretable than the compliers of the local
average treatment effect \citep{imbens1994}, and under heterogeneous effects it answers a different question
than the ATT for the full treated population. We handle this directly. The gap between the two is the
composition gap $\Gamma$, unidentified because it depends on the effects of the discarded cohorts, so we
always report both honest intervals with the breakdown value $\Gamma^\ast$, and the reader sees the
full-population target alongside the credible one and how far the discarded cohorts' effects must differ to
overturn the
reading, and wherever the retained set admits a substantive description, in the application of
Section~\ref{sec:app}, the counties whose oil-and-gas onset did not coincide with the mid-2000s housing
run-up, lead with that description, which converts a cohort selected for a convenient statistical
property into a subpopulation a reader can name, characterizing it by its covariates the way applied work
characterizes the compliers of an instrumental-variables design \citep{abadie2003}.

The construction uses staggered timing only to supply the groups. What it requires is a partition of the
treated population into groups whose parallel trends can each be assessed. Adoption cohorts are the
natural groups, and their credibility heterogeneity comes for free from the variation in when and why
units adopt. In a non-staggered design the same construction applies with groups defined by covariates or
geography, provided the groups differ in credibility and there are enough pre-periods to estimate each
group's pre-trend rather than chase noise. The groups should accordingly be subpopulations large enough
to assess, not individual units, whose pre-trends are too noisy to select on without reviving the
pre-test bias in its most severe form.

\subsection{Selection and estimation}\label{sec:select}

The selection rule $R$ maps the data to the set $S$, and its form fixes both which subpopulation
$\theta_S$ describes and the inferential guarantees available. It takes two forms, matching the two
regimes of Section~\ref{sec:setup}. Under ex-ante selection, $S$ is fixed before estimation from
pre-determined covariates, institutional knowledge of which cohorts were plausibly confounded, or an
independent data split. Under data-driven selection, credibility is read from the estimated
pre-trends, and a cohort enters $S$ when its estimated pre-trend is close to flat,
\begin{equation}
S \;=\; \big\{\, g : \max_{e<0}\, \lvert \hat\beta_{g,\mathrm{pre}}(e)\rvert \le c \,\big\},
\label{eq:rule}
\end{equation}
for a threshold $c$.

This is the direct reading of parallel trends. The assumption is that the treated-comparison difference
would have stayed constant absent treatment, so a cohort is credible when its pre-treatment difference
is close to constant, and the screen keeps exactly those. It is posed in the same currency as the honest
inference of the next subsection, which bounds how far the post-treatment difference could have wandered
from flat, so selection and inference test one object, the level of the parallel-trends
violation.\footnote{A researcher willing to trust that a straight pre-trend would have continued linearly
can relax flatness to approximate linearity, screening instead on the pre-trend's curvature (its second
difference) and pairing it with the smoothness sensitivity class of \citet{rambachan2023} and a
linear-extrapolation estimator \citep{armstrong2018}. This buys the steep-but-straight cohorts that
flatness discards, at the cost of the extrapolation assumption. The principle is the same in either
currency. The screen should be posed in whichever functional the sensitivity class bounds, flatness with
the level bound and curvature with the smoothness class, since level and curvature rank cohorts
differently and select different sets.}

The threshold trades scope
against credibility, with a lax $c$ retaining more cohorts but admitting larger residual violations and a
strict $c$ the reverse. Because rule~\eqref{eq:rule} reads $S$ off the noisy estimates $\hat\beta_{\mathrm{pre}}$, the selected set
is random and dependent on the estimation sample, which is the source of the pre-test bias below and of the
inference subtleties that follow it.

We do not recommend committing to a single $c$. The sensitivity interval of
Section~\ref{sec:approach} is applied to whichever set $c$ selects, and its bound $M$ must dominate that
set's residual post-treatment violation, so the effect of $c$ on the required $M$ runs through how
far screening on flat pre-trends also flattens post-trends. When pre-trends are informative about
post-treatment violations, a stricter $c$ yields a cleaner set that a smaller $M$ suffices to cover, while
a laxer $c$ buys scope at the cost of a larger $M$, so the threshold trades scope for precision rather than
credibility for nothing. When they are not, tightening $c$ need not shrink the post-treatment violation and
the required $M$ can be nonmonotone in $c$. Either way the honest object to report is the estimand together
with its interval traced as a function of $c$, the selection analogue of the breakdown curve we report for
$M$. A reader then sees the scope-credibility-precision frontier directly and need not trust a single
cutoff or a presumed monotone mapping. When a
single value is nonetheless wanted, we recommend calibrating it to negligibility rather than to
significance, setting $c$ to the largest pre-trend that would bias $\theta_S$ by less than a stated
tolerance, in the spirit of the non-inferiority tests of \citet{bilinski2020}, rather than to a critical
value of a pre-trends test, which would reinherit the low power that motivates the paper. The
application of Section~\ref{sec:app} reports such a path.

The estimator $\hat\theta_S$ simply reweights, over the selected cohorts, the sample group-time effects
of whichever heterogeneity-robust estimator one uses, and nothing below depends on that choice.

\begin{proposition}[Consistency and efficiency]\label{prop:cons}
Under ex-ante selection and the regularity conditions of the underlying group-time estimator,
$\hat\theta_S$ is consistent for $\theta_S$. If in addition parallel trends holds for all cohorts,
$\hat\theta_S$ is consistent for the ATT when either treatment effects are homogeneous or $S$ comprises all
cohorts, and otherwise for the selected-cohort LATT, which differs from the ATT by the gap between the
selected- and full-population weighted averages of the cohort effects. In either case there is an
efficiency loss relative to estimators that use all cohorts.
\end{proposition}

The consistency clause of Proposition~\ref{prop:cons} invites a sharper question. When does narrowing the
target cost nothing, in that the LATT and the ATT are the same object? Writing $\theta_g:=\mathrm{ATT}(g,\cdot)$,
$\mathcal G$ for the full set of treated cohorts, and $\theta_{\mathcal G}$ for the ATT, the two coincide for
a fixed selected set $S$, $\theta_S=\theta_{\mathcal G}$, exactly when the retained cohorts' weighted-average
effect equals the full population's,
$\sum_{g\in S}w_g\theta_g/\sum_{g\in S}w_g=\sum_{g\in\mathcal G}w_g\theta_g/\sum_{g\in\mathcal G}w_g$. In
particular when no cohort is dropped, when effects are homogeneous, and more generally whenever selection is
uninformative about effect size. The gap $\theta_S-\theta_{\mathcal G}$ is the composition gap $\Gamma$ that
Proposition~\ref{prop:generalK} isolates in Appendix~\ref{app:mse}, the LATT's counterpart of the selection
that separates a local from an average treatment effect \citep{imbens1994,li2018}. Target equivalence and estimator convergence are
distinct. Under data-driven selection the estimators can still converge, when no cohort is detectable and the
whole population is retained, by Proposition~\ref{prop:pretest}, even where the targets would differ for any
strict selection, while a strict but ignorable selection leaves the targets equal yet the finite-sample
estimators differ by the efficiency cost of Proposition~\ref{prop:cons}.

Under data-driven selection the rule~\eqref{eq:rule} couples $S$ to the estimation errors. A confounded
cohort whose pre-trend noise happens to fall below $c$ is admitted, and its differential trend enters
$\hat\theta_S$. This is a pre-test selection bias in the sense of \citet{roth2022pretrends}, with a
definite source and a definite fate.

\begin{proposition}[Pre-test bias]\label{prop:pretest}
Under data-driven selection, $\hat\theta_S$ carries a finite-sample bias whose leading term is the weighted
average of the admitted confounded cohorts' post-treatment differential trends, times their probability of
admission. The exact bias also reflects the random weight normalization on the selected set and, when pre-
and post-treatment estimation errors are correlated, a term from conditioning on the selection. Whether the
leading term vanishes turns on detectability. Call a confounded cohort detectable if its
population pre-trend statistic exceeds the threshold, $\max_{e<0}\lvert\beta_{g,\mathrm{pre}}(e)\rvert>c$,
and undetectable if it is flat in the pre-period, $\max_{e<0}\lvert\beta_{g,\mathrm{pre}}(e)\rvert\le c$,
yet violates parallel trends in the post-period. As per-cohort information grows, a detectable cohort's
admission probability tends to zero and its contribution to the bias vanishes, whereas an undetectable
cohort's admission probability tends to one and its contribution persists at every sample size, while a
cohort exactly at the threshold is admitted with an interior probability. Hence
$\hat\theta_S$ is consistent for $\theta_S$ whenever the selected aggregate's residual violation vanishes,
which holds in particular if every confounded cohort is detectable, though detectability of each is
sufficient rather than necessary, since cohort-level violations may cancel in the aggregate. Otherwise
$\hat\theta_S$ retains an asymptotic bias equal to the residual violation of the undetectable cohorts, the
identification gap that the honest inference of the next subsection bounds.
\end{proposition}

\section{Honest inference}\label{sec:honest}

A flat pre-trend does not guarantee a flat post-trend. Selection removes gross violations but not subtle
residual ones. We therefore pair the point estimate with the sensitivity analysis of \citet{rambachan2023},
applied to the selected-cohort aggregate. It imposes a restriction $\delta_{S,\mathrm{post}}\in\Delta$ on
the residual differential trend and reports the implied confidence set for $\theta_S$. To match the
flatness screen we use the level bound
$\Delta^{\mathrm{Level}}(M)=\{\delta:\lvert\delta_{\mathrm{post}}(e)\rvert\le M\ \forall e\}$, which asks how
far the treated-comparison difference could have drifted from its flat pre-treatment level. The estimator
is then the reweighted post-treatment average, with no extrapolation, and its fixed-length confidence
interval (FLCI) is near-optimal \citep{armstrong2018}. We use this imported construction throughout. For a
scalar target estimated with standard error $s$ under a worst-case bias bound $B$, the near-optimal FLCI of
\citet{armstrong2018} has half-width
\begin{equation}\label{eq:flci-hw}
h(B,s)\;=\;s\,\mathrm{cv}_\alpha\!\big(B/s\big),\qquad
\mathrm{cv}_\alpha(t)=\text{the }(1-\alpha)\text{ quantile of }\lvert\mathcal N(t,1)\rvert,
\end{equation}
with $\mathrm{cv}_\alpha(0)=z_{1-\alpha/2}$ the pure-noise critical value and
$\mathrm{cv}_\alpha(t)\uparrow t+z_{1-\alpha}$ as the bias comes to dominate. The bound $M$ is not estimated but is the axis of a
sensitivity analysis. The pre-period identifies $\delta_{\mathrm{pre}}$ and not $\delta_{\mathrm{post}}$, so no
screen can certify a value of $M$, and the honest object is not a single interval but the map from $M$ to the
interval it implies, together with the breakdown value $M^\ast$, the smallest post-treatment violation under
which the stated conclusion no longer follows. This is the level-bound analogue of the breakdown reporting of
\citet{rambachan2023}. The honest object is not a claim that the retained cohorts' post-trends are flat, which no pre-trend can
establish, but the map of how flat they must be for the reading to stand. Because $M$ and $M^\ast$ are in
the outcome's own units, $M^\ast$ is directly interpretable. We benchmark it two ways: against the retained
cohorts' own pre-trends, which the screen holds below $c$ in the same units, and against the effect at
stake. An $M^\ast$ several times the retained cohorts' visible pre-trends signals a conclusion surviving
violations far larger than anything the pre-period displays. Section~\ref{sec:app} reports both benchmarks.

Validity depends on the selection regime. Under ex-ante selection $S$ is independent of $\hat\beta$, and the
uniform validity of \citet{rambachan2023} applies verbatim. Under data-driven selection $S$ is random, and
two issues arise. The first is post-selection inference. It is asymptotically harmless when cohorts are well
separated from the threshold.

\begin{remark}[Pointwise oracle coverage under separation]\label{rem:oracle}
A consistency argument already secures coverage under a strong condition. If the selection statistic is
consistent and every cohort's population value lies a fixed distance from the threshold $c$ (a separation
condition), and the population credible set's residual violation satisfies
$\delta_{S^\ast,\mathrm{post}}\in\Delta^{\mathrm{Level}}(M)$, then $P(\hat S=S^\ast)\to1$ and the FLCI on the
estimated set $\hat S$ inherits the oracle coverage of the interval on the population set $S^\ast$, nominal
by the fixed-set validity noted above. This oracle guarantee has two limitations, both removed by the carved construction
below. The separation condition is untestable
and nearly restates the parallel-trends question, relocating the identifying content rather than supplying
it, and the coverage is only pointwise, since a post-selection distribution cannot be estimated uniformly
consistently \citep{leeb2005,leeb2008}, the obstruction confined to the local-to-threshold configurations
where the selection statistic drifts toward the boundary at the sampling rate.
\end{remark}

A researcher unwilling to assume separation has two options. One secures validity through independence,
selecting on pre-determined covariates or an independent sample split, at the cost of the efficiency the
split forgoes. The other conditions on the selection and recovers that efficiency. We develop it now,
because it delivers the uniform guarantee the consistency argument lacks. The result rests on the following
conditions and on the classical post-selection lemma, both collected here so the construction is
self-contained.

\begin{assumption}[Carved-inference conditions]\label{ass:carve}
In addition to Assumption~\ref{ass:main}, the finite-sample carved inference maintains. (i) the Gaussian
model $\hat\beta\sim\mathcal N(\beta,\Sigma)$ with $\Sigma$ known, the estimated-$\Sigma$ case treated in
Remark~\ref{rem:sigmahat}. (ii) the data-driven flatness screen $R$ of \eqref{eq:rule}, randomized at a
scale $\gamma\ge0$, and (iii) conditioning on the realized selection event together with the active
pre-period constraints and their signs, for each retained cohort the constraints that bind, for each
dropped cohort a breaching coordinate and its sign.
\end{assumption}

\begin{lemma}[Polyhedral lemma, \citealp{lee2016}]\label{lem:poly}
Let $y\sim\mathcal N(\mu,\Sigma)$ and fix a contrast $\eta$. On any polyhedral event $\{Ay\le b\}$,
decompose $Ay$ into its component along $\eta'y$ and a remainder $r$ that is independent of $\eta'y$.
Conditional on the event $\{Ay\le b\}$ and on the realized remainder $r$, the event confines $\eta'y$ to an
interval $[\mathcal V^-,\mathcal V^+]$ whose endpoints are computable in closed form from $(A,b,\Sigma,\eta)$
and $r$, and its law is $\mathcal N(\eta'\mu,\eta'\Sigma\eta)$ truncated to $[\mathcal V^-,\mathcal V^+]$. The
truncated-normal pivot $F^{[\mathcal V^-,\mathcal V^+]}_{\eta'\mu,\,\eta'\Sigma\eta}(\eta'y)$ is therefore
$\mathrm{Uniform}[0,1]$ conditional on the event and $r$, and, since this holds for every value of $r$,
conditional on the event alone by iterated expectations.
\end{lemma}

The construction randomizes the screen and conditions on its outcome. In the reduced-form model
$\hat\beta\sim\mathcal N(\beta,\Sigma)$ with $\Sigma$ known, draw independent noise
$\omega\sim\mathcal N(0,\gamma\,\Sigma_{\mathrm{pre}})$ for a scale $\gamma\ge0$, and select on the noised
pre-trends, $\hat S=\{g:\max_e\lvert\hat\beta_{g,\mathrm{pre}}(e)+\omega_{g}(e)\rvert\le c\}$, while
estimating the aggregate $\theta_S=\ell_S'\beta_{\mathrm{post}}$ from the unrandomized
$\hat\theta_S=\ell_S'\hat\beta_{\mathrm{post}}$, where $\ell_S$ carries the reweighting onto $S$. The
target is the reduced-form aggregate. Under parallel trends for $S$ it is the causal $\theta_S$, and
otherwise it is that quantity plus the residual violation the level bound handles below. Selection
uncertainty and identification are thus carried by separate instruments.

\begin{theorem}[Uniformly valid carved inference]\label{thm:carve}
Maintain Assumption~\ref{ass:carve} and fix $\gamma\ge0$. Conditioning on the selection event $\{\hat S=S\}$
together with the active flatness-screen coordinates and their signs renders the event a polyhedron in the
Gaussian vector even with several pre-periods per cohort. Then, conditional on $\{\hat S=S\}$ and the
Gaussian remainder $r$ orthogonal to $\hat\theta_S$, the estimate has a truncated Gaussian law,
\[
\hat\theta_S\mid\{\hat S=S\},\,r\ \sim\ \mathcal{TN}\big(\theta_S,\ s_S^2,\ [\mathcal V^-,\mathcal V^+]\big),
\qquad s_S^2=\ell_S'\Sigma_{\mathrm{post}}\,\ell_S,
\]
whose truncation limits $[\mathcal V^-,\mathcal V^+]$ are computable in closed form from $\Sigma$, $\gamma$,
and the realized data by the polyhedral lemma of \citet{lee2016}. Inverting the
truncated-normal pivot $F_{\theta_S,s_S^2}^{[\mathcal V^-,\mathcal V^+]}(\hat\theta_S)$ yields an interval
$\mathcal C_{1-\alpha}(S)$ whose pivot is uniform for every realized $r$, hence uniform conditional on
$\{\hat S=S\}$ by iterated expectations. This gives exact conditional coverage
$\Pr(\theta_S\in\mathcal C_{1-\alpha}(S)\mid\hat S=S)=1-\alpha$ for every $S\neq\varnothing$ and every
$\beta$, and hence exact unconditional coverage $\Pr(\theta_S\in\mathcal C_{1-\alpha}(\hat S)\mid\hat
S\neq\varnothing)=1-\alpha$, uniformly in $\beta$ and without any separation condition.
\end{theorem}

\begin{proof}
For a retained cohort the screen $\max_e\lvert\hat\beta_{g,\mathrm{pre}}(e)+\omega_g\rvert\le c$ is a box, an
intersection of half-spaces. For a dropped cohort the complementary event $\max_e\lvert\cdot\rvert>c$ is a
union of half-spaces and is not itself convex. Conditioning on a breaching coordinate $e_g^\ast$ and its
sign selects the single half-space $\pm(\hat\beta_{g,\mathrm{pre}}(e_g^\ast)+\omega_g)>c$, and when several
coordinates breach we condition on the full active set. The conditioned selection event is thus an intersection
of half-spaces, a polyhedron in $(\hat\beta,\omega)$, for any number of pre-periods per cohort. This
conditioning only refines the event, so the conditional pivot stays exact and unconditional coverage is
recovered by averaging over the finer partition as well as over $S$. The
pair $(\hat\theta_S,\,\text{selection})$ is jointly Gaussian, and decomposing the constraint vector into
its component along $\hat\theta_S$ and an orthogonal remainder that is independent of $\hat\theta_S$, the
event fixes the remainder and confines $\hat\theta_S$ to an interval $[\mathcal V^-,\mathcal V^+]$ whose
endpoints are the tightest of the resulting linear bounds. This is Lemma~\ref{lem:poly} applied to the
conditioned event. Hence the conditional law of $\hat\theta_S$ is $\mathcal N(\theta_S,s_S^2)$ truncated to
that interval, the pivot is uniform on $[0,1]$ conditional on $\hat S=S$, and inverting it gives an
interval with conditional coverage $1-\alpha$ for every $S$ and every $\beta$. Averaging over $S$ gives
the same unconditional coverage, and since the finite-sample pivot is exact for known $\Sigma$ regardless
of where $\beta$ sits relative to the threshold, coverage is uniform. The endpoint claims in $\gamma$ hold
because $\gamma=0$ leaves the selection a deterministic function of $\hat\beta_{\mathrm{pre}}$, while as
$\gamma\to\infty$ the correlation between the selection variable and $\hat\theta_S$ vanishes, so
$[\mathcal V^-,\mathcal V^+]$ expands to the real line and the truncated law returns to
$\mathcal N(\theta_S,s_S^2)$.
\end{proof}

\begin{remark}[The truncation limits and the role of $\gamma$]\label{rem:carve-detail}
The limits $[\mathcal V^-,\mathcal V^+]$ are the intersection, over retained cohorts of the two-sided
constraints $\lvert\hat\beta_{g,\mathrm{pre}}+\omega_g\rvert\le c$ and over dropped cohorts of the one-sided
active constraint $\pm(\hat\beta_{g,\mathrm{pre}}(e_g^\ast)+\omega_g)>c$, of the intervals these place on
$\hat\theta_S$ once $r$ is held fixed. With a single pre-period per cohort, and when the dropped cohorts'
screening statistics are uncorrelated with $\hat\theta_S$, the retained-cohort constraints alone bind; in
general the dropped-cohort constraints also enter and are retained. The scale $\gamma$ is a power
parameter, not a validity parameter. At $\gamma=0$ the interval is fully conditional; intermediate $\gamma$
carves, recovering efficiency that full conditioning spends while retaining exact validity; and as
$\gamma\to\infty$ the noise decorrelates the selection from $\hat\theta_S$ but sends every retention
probability to zero, so $\hat S=\varnothing$ with probability approaching one. The useful regime is
intermediate $\gamma$, and on the empty event the procedure reports the ATT interval.
\end{remark}

Theorem~\ref{thm:carve} upgrades the pointwise oracle coverage of Remark~\ref{rem:oracle} to a
uniform guarantee, precisely across the local-to-threshold configurations that obstruct it. The guarantee is
exact for known $\Sigma$ and degrades to first-order validity when $\Sigma$ is consistently estimated
(Remark~\ref{rem:sigmahat}). That is the paper's single scope qualification, and we do not restate it
elsewhere. It does so for
the sampling and selection uncertainty in the reduced-form aggregate. The identification gap remains with
the level bound, which we now fold in to reach the causal target.

\begin{corollary}[Honest causal interval under data-driven selection]\label{cor:causal}
Let $\theta_{\hat S}^{\mathrm c}=\ell_{\hat S}'\tau_{\mathrm{post}}$ be the causal target on the selected set and
$B_{\hat S}=\sup_{\delta_{\mathrm{post}}\in\Delta^{\mathrm{Level}}(M)}\lvert\ell_{\hat S}'\delta_{\mathrm{post}}\rvert$
the worst-case residual bias the level bound allows, so that
$\lvert\theta_{\hat S}-\theta_{\hat S}^{\mathrm c}\rvert\le B_{\hat S}$ whenever
$\delta_{\hat S,\mathrm{post}}\in\Delta^{\mathrm{Level}}(M)$. For a single-period or convex-averaging functional
$a(\cdot)$, $B_{\hat S}=M$. Then the widened carved interval
\[
\mathcal C^{M}_{1-\alpha}(\hat S)\ :=\ \big[\,\inf\mathcal C_{1-\alpha}(\hat S)-B_{\hat S},\
\sup\mathcal C_{1-\alpha}(\hat S)+B_{\hat S}\,\big]
\]
covers the causal target, $\Pr\big(\theta_{\hat S}^{\mathrm c}\in\mathcal C^{M}_{1-\alpha}(\hat S)\big)\ge1-\alpha$,
uniformly in $\beta$ and without a separation condition, whenever
$\delta_{\hat S,\mathrm{post}}\in\Delta^{\mathrm{Level}}(M)$.
\end{corollary}

\begin{proof}
See Appendix~\ref{app:proofs}.
\end{proof}

The widened interval $\mathcal C^{M}_{1-\alpha}(\hat S)$ is conservative. It superposes the
exact carved interval for the reduced-form $\theta_{\hat S}$ on the worst-case additive bias $B_{\hat S}$ rather
than folding sampling, selection, and identification uncertainty into a single length-optimal object, and is
accordingly not the near-optimal fixed-length interval of \citet{armstrong2018}. That construction and its
optimality belong to the fixed-set case of Theorem~\ref{thm:width}, where selection does not randomize the target
and the bias enters only through the standard error. Under data-driven selection the superposition buys a single
interval valid against both the selection and the residual violation, at a modest excess length. The carved
interval's exact coverage is the content of Theorem~\ref{thm:carve}, and the multi-cohort panel study of
Section~\ref{sec:calib} exhibits the understatement of naive uncertainty under estimated selection that
carving is built to correct.

The second issue is that the bound $M$ must accommodate a random selected set. The FLCI at $M$ is valid
only if the selected set's residual violation lies in $\Delta^{\mathrm{Level}}(M)$. Since that violation
is itself random, unconditional coverage requires $M$ to dominate it with high probability.

\begin{remark}[Calibration to the random selected set]\label{rem:calib}
Under data-driven selection the selected aggregate's residual violation is random, so the level bound $M$
may fail to dominate it. Write $\pi_M=\Pr\big(\delta_{\hat S,\mathrm{post}}\notin\Delta^{\mathrm{Level}}(M)\big)$
for the probability that the realized violation escapes the bound. Then unconditional coverage is at least
$1-\alpha-\pi_M$. The carved interval covers the reduced-form $\theta_{\hat S}$ with probability $1-\alpha$
by Theorem~\ref{thm:carve}, and its widening covers the causal target on the event that the bound holds, so
$\{\theta_{\hat S}\in\mathcal C_{1-\alpha}(\hat S)\}\cap\{\delta_{\hat S,\mathrm{post}}\in
\Delta^{\mathrm{Level}}(M)\}\subseteq\{\theta_{\hat S}^{\mathrm c}\in\mathcal C^{M}_{1-\alpha}(\hat S)\}$ and a
union bound gives the claim. Coverage thus approaches nominal only as $\pi_M\to0$. Calibrating $M$
to the mean residual violation typically leaves $\pi_M$ bounded away from zero and undercovers. Taking $M$ at
an upper quantile of the residual-violation distribution makes $\pi_M$ small but not zero, so it delivers
coverage $1-\alpha-\pi_M$ rather than exactly $1-\alpha$ unless the error budget is widened to absorb
$\pi_M$. That distribution is unidentified, which is
why the reported object is the breakdown curve rather than a single $M$.
\end{remark}

The practical implication is not a recipe for $M$ but a warning against one. Because the violation $M$ must
dominate is itself random, no single $M$ is at once known and unconditionally valid. Calibrating $M$ to the
mean residual violation undercovers, since the realized violation exceeds it in a nontrivial fraction of
samples, while the $M$ that would deliver a target coverage sits at an upper quantile of a distribution the
design cannot identify. This is why the reported object is the breakdown curve, the map from $M$ to its
interval, and the value $M^\ast$, read conservatively, rather than a point calibration of $M$. Each interval
on the curve is honest at its own $M$ by Corollary~\ref{cor:causal}, and Remark~\ref{rem:calib} adds only that
the curve must be read as a whole, since a bound calibrated to the typical selected set would place a falsely
narrow interval at the small-$M$ end.

Finally, we caution against one otherwise-natural choice of $\Delta$. The relative-magnitudes restriction
$\Delta^{RM}(\bar M)$ of \citet{rambachan2023}, following \citet{manski2018}, bounds post-treatment
violations by $\bar M$ times the observed maximal pre-treatment violation, so its identified-set half-width
is proportional to the observed pre-trend. When pre-trends are uninformative, they are small while the true
violation need not be, so the identified set narrows toward a point even as the bias does not. Coverage then
fails not from undercoverage within the maintained class but because the truth can lie outside it, the small
observed pre-trend implying a relative-magnitudes class that excludes the true post-treatment violation. The
fixed bound $\Delta^{\mathrm{Level}}(M)$, set by the
researcher rather than by the observed pre-trends, avoids tying the admissible violation to the pre-trend,
though it too is honest only when its chosen $M$ contains the true violation.

\begin{remark}[Estimated covariance]\label{rem:sigmahat}
Theorem~\ref{thm:carve} treats $\Sigma$ as known. In practice one plugs in a consistent $\hat\Sigma$, which
enters twice, in the aggregate standard error $s_S$ and in the selection geometry
$[\mathcal V^-,\mathcal V^+]$ and carving scale, so it is more consequential here than in a fixed-target
sensitivity analysis, where $\Sigma$ enters only the standard error. If $\hat\Sigma\xrightarrow{p}\Sigma$ and
the separation condition of Remark~\ref{rem:oracle} holds, the truncation limits and scale are continuous in
$\hat\Sigma$ and the selection event is correctly classified with probability approaching one, so by Slutsky
the pivot converges to $\mathrm{Uniform}[0,1]$. The known-$\Sigma$ exactness degrades to first-order valid
coverage, which the panel simulation of Section~\ref{sec:calib}, estimating the covariance from the data,
bears out. Separation is not cosmetic. Within a $\hat\Sigma$-sampling neighbourhood of the
threshold a cohort's inclusion is itself estimated, the estimated-$\Sigma$ shadow of the Leeb--P\"otscher
obstruction, so a researcher unwilling to assume it should hold the carving scale $\gamma$ away from its
degenerate limit or use the independent split.
\end{remark}

\section{When the trade pays}\label{sec:dominance}

The deliverable of this section is a rule the practitioner reports alongside both honest intervals. Prefer
the credible-subpopulation LATT when narrowing lowers risk, and the ATT otherwise. Three forces set whether
narrowing helps: the violation the screen removes, the variance that dropping a cohort costs, and the
heterogeneity between the retained subpopulation and the full one. The sign of the advantage turns on their
balance, not on the informativeness of the pre-trends, which governs only how much of a fixed-sign
advantage is realized. The mean-squared-error account of these forces supplies the intuition and the
explicit thresholds, the two-cohort law $V^2>\gamma^2+2\sigma^2$, the estimable diagnostic
$D^2>\Delta\mathrm{Var}$, and the breakdown value $\Gamma^\ast$ for the unidentified composition gap. It is
developed in Appendix~\ref{app:mse}, whose results we cite as needed. Here we give the operative test, the
width comparison of the honest intervals themselves, since the intervals, not the point estimators, are the
recommended output.

The comparison of honest intervals can be made exact in a two-cohort model, in the informative regime where
the screen discards the confounded cohort with probability approaching one, so that the LATT selects the
clean cohort alone. This is the regime in which the point-estimate advantage attains its ceiling
(Appendix~\ref{app:mse}), and it is where the method is meant to help. Both procedures report the
level-bound fixed-length interval \eqref{eq:flci-hw}. The LATT, using the clean cohort alone, is unbiased
with standard error $\sigma$ and a residual-violation bound $M_L$, giving half-width
$\sigma\,\mathrm{cv}_\alpha(M_L/\sigma)$. The ATT estimator, using both cohorts, has standard error
$\sigma/\sqrt2$ but must set its bound to cover its own aggregate violation, $M_A=V/2$, giving half-width
$(\sigma/\sqrt2)\,\mathrm{cv}_\alpha(V/(\sqrt2\,\sigma))$. Under homogeneous effects, so that the
clean-cohort LATT and the full ATT share a common target, each interval covers that target at its
nominal level by construction, since each bound dominates its own aggregate's violation, so the honest
comparison is not one of coverage but of expected length at equal coverage. Under heterogeneity the two
intervals honestly cover different objects, and the comparison is one of length only once the composition
gap is folded into the LATT interval evaluated as inference for the ATT.

We state the comparison in general, since it holds at
any cohort count and isolates the single quantity, the violation carried by the dropped
cohorts, that decides it. Write the ATT estimator as
$\hat\theta_{\mathcal G}=\ell_{\mathcal G}'\hat\beta_{\mathrm{post}}$ with aggregate standard
error $s_{\mathcal G}^2=\ell_{\mathcal G}'\Sigma_{\mathrm{post}}\ell_{\mathcal G}$ and honest
bias bound $B_{\mathcal G}$ dominating the full set's violation
$\lvert\ell_{\mathcal G}'\delta_{\mathrm{post}}\rvert$, and the LATT estimator as
$\hat\theta_S=\ell_S'\hat\beta_{\mathrm{post}}$ with carved standard error $s_S$ and residual
bound $B_S$. Both honest intervals are the fixed-length construction \eqref{eq:flci-hw}.

\begin{theorem}[Width dominance]\label{thm:width}
Suppose both honest intervals are the fixed-length construction \eqref{eq:flci-hw}, the LATT interval under
ex-ante selection, equivalently the \citet{rambachan2023} interval on the retained set, with half-width
$h_S=h(B_S,s_S)$, and the ATT interval $h_{\mathcal G}=h(B_{\mathcal G},s_{\mathcal G})$. Then
\[
h_S<h_{\mathcal G}\quad\Longleftrightarrow\quad
s_S\,\mathrm{cv}_\alpha\!\big(B_S/s_S\big)\;<\;
s_{\mathcal G}\,\mathrm{cv}_\alpha\!\big(B_{\mathcal G}/s_{\mathcal G}\big).
\]
Holding $(s_{\mathcal G},s_S,B_S)$ fixed, if $s_S\ge s_{\mathcal G}$ there is a unique threshold
\[
B_{\mathcal G}^\ast\;=\;s_{\mathcal G}\,\mathrm{cv}_\alpha^{-1}\!\Big(\tfrac{s_S}{s_{\mathcal G}}\,
\mathrm{cv}_\alpha\!\big(B_S/s_S\big)\Big)
\]
such that the credible-subpopulation interval is strictly shorter than the ATT interval if and
only if $B_{\mathcal G}>B_{\mathcal G}^\ast$. If instead $s_S<s_{\mathcal G}$, the LATT interval can be
shorter already at $B_{\mathcal G}=0$ and no positive crossing threshold exists. Since $B_{\mathcal G}$
carries the dropped cohorts' aggregate violation $B_{\mathcal F}$ while $B_S$ does not, this is a threshold
on $B_{\mathcal F}$, and when restricting to $S$ raises the aggregate variance,
$s_S>s_{\mathcal G}$, the threshold is strictly positive, so a minimum dropped violation is
required before bias relief outweighs the variance cost.
\end{theorem}

\begin{proof}
See Appendix~\ref{app:proofs}.
\end{proof}

Under data-driven selection the LATT interval is instead the carved interval of Theorem~\ref{thm:carve},
whose truncated-normal inversion gives a data-dependent, generally asymmetric length rather than the
fixed half-width $h(B_S,s_S)$. Here the width comparison does not extend, and for a definite reason. The
carved interval inverts a truncated-normal pivot, and by \citet{kivaranovicleeb2021} such an interval has
\emph{infinite} conditional expected length, so an expected-length threshold is not merely unproven but
ill-posed; Proposition~\ref{prop:nolength} records this. The honest data-driven comparison is therefore by
coverage, both intervals valid by Theorem~\ref{thm:carve} and Corollary~\ref{cor:causal}, together with a
robust length summary such as the median, which the simulation of Section~\ref{sec:calib} reports. The
fixed-length statement of Theorem~\ref{thm:width} is accordingly stated for ex-ante selection, where the
intervals are the fixed-length construction \eqref{eq:flci-hw} and the threshold is exact.

\begin{proposition}[No expected-length dominance under data-driven selection]\label{prop:nolength}
Maintain Assumption~\ref{ass:carve}. Under data-driven selection the carved interval
$\mathcal C_{1-\alpha}(\hat S)$ of Theorem~\ref{thm:carve} is the polyhedral inversion of the
truncated-normal pivot $F_{\theta_S,s_S^2}^{[\mathcal V^-,\mathcal V^+]}$, and for every $\beta$ and every
$S\neq\varnothing$,
\[
\mathbb E\big[\,\lvert\mathcal C_{1-\alpha}(S)\rvert\ \big|\ \hat S=S\,\big]\;=\;\infty.
\]
Consequently no expected-length ordering between the carved LATT interval and the ATT interval exists, and
the width comparison of Theorem~\ref{thm:width} is well-posed only for the ex-ante fixed-length intervals.
The randomization scale does not remove this, since every finite $\gamma$ conditions on $\{\hat S=S\}$ and
the active screen coordinates, leaving the pivot truncated normal.
\end{proposition}

\begin{proof}
See Appendix~\ref{app:proofs}.
\end{proof}

Specializing to equal weights makes the threshold explicit. Let $K$ equally weighted cohorts with common
per-cohort post-treatment variance $\sigma^2$ and
independent estimates comprise $k$ clean cohorts ($V_g=0$), retained, and $K-k$ confounded cohorts
(violation $V$), dropped in the informative limit, with the retained set truly flat, $B_S=0$.
Writing $f=(K-k)/K$, so that $s_{\mathcal G}=\sigma/\sqrt K$, $s_S=\sigma/\sqrt k$, and
$B_{\mathcal G}=fV$,
\[
h_S<h_{\mathcal G}\quad\Longleftrightarrow\quad
\mathrm{cv}_\alpha\!\Big(\tfrac{fV\sqrt K}{\sigma}\Big)>\sqrt{\tfrac{K}{k}}\,z_{1-\alpha/2}
\quad\Longleftrightarrow\quad
V>V^\ast:=\frac{\sigma}{f\sqrt K}\,
\mathrm{cv}_\alpha^{-1}\!\Big(\sqrt{\tfrac{K}{k}}\,z_{1-\alpha/2}\Big),
\]
equivalently $B_{\mathcal F}=fV>B_{\mathcal F}^\ast=(\sigma/\sqrt K)\,
\mathrm{cv}_\alpha^{-1}\big(\sqrt{K/k}\,z_{1-\alpha/2}\big)$. At $K=2$, $k=1$ this is $V>1.59\,\sigma$
(at $\alpha=0.05$).

Two features carry the economics. The width threshold $V>1.59\,\sigma$ exceeds the mean-squared-error
threshold $V>\sqrt2\,\sigma\approx1.41\,\sigma$ of Remark~\ref{rem:inform}, because the fixed-length
interval charges for the standard-error inflation of dropping a cohort, not only for its variance, and
coverage is never the margin, since both procedures are honest, the ATT pays for honesty in a width that
grows linearly in $V$, whereas the LATT pays a fixed $2\sigma\,z_{1-\alpha/2}$ set by its variance alone,
the price-of-honesty reading of Figure~\ref{fig:layer2}. Allowing the retained set a residual bound
$M_L>0$ rather than exact flatness, the LATT interval carries standard error $\sigma$ and bound $M_L$
against the ATT's $\sigma/\sqrt2$ and $M_A=V/2$, so by \eqref{eq:flci-hw} it is shorter if and only if
$\mathrm{cv}_\alpha(M_L/\sigma)<\tfrac1{\sqrt2}\,\mathrm{cv}_\alpha(V/(\sqrt2\,\sigma))$. A residual
$M_L>0$, and heterogeneity $\gamma\neq0$ when the target is the ATT (which adds $\lvert\gamma\rvert/2$ to
the LATT's effective bias), both raise the threshold, in the direction the point analysis already
indicated.

The width comparison assembles into a rule the practitioner can run. Report both honest
intervals, for the ATT and for the LATT, and prefer whichever is shorter. Because each interval covers its
own target whatever the data say, ranking them can misfire on efficiency but never on coverage.

\begin{proposition}[A feasible and honest switching rule]\label{prop:switch}
Suppose the researcher reports the honest level-bound interval for both the ATT and the LATT, with
half-widths $h_{\mathcal G}=h(B_{\mathcal G},s_{\mathcal G})$ and $h_S=h(B_S,s_S)$ of \eqref{eq:flci-hw},
together with the rule. Prefer the LATT when its honest interval is strictly shorter, $h_S<h_{\mathcal G}$,
and the ATT otherwise. Then the following hold.
\begin{enumerate}
\item[(i)] (Honesty is unconditional.) Each reported interval covers its own target at level
$1-\alpha$ for every $\beta$, by the validity of the level-bound interval (carved, under data-driven
selection, by Theorem~\ref{thm:carve}). Since both are always reported, the rule, however it
decides, cannot affect coverage. It ranks the two honest intervals and is an efficiency device only.
\item[(ii)] (Feasibility.) The rule is computable from $\hat\Sigma$ and the chosen bounds alone. By
Theorem~\ref{thm:width}, $h_S<h_{\mathcal G}$ if and only if the ATT's honest bound $B_{\mathcal G}$, which
carries the dropped cohorts' aggregate violation $B_{\mathcal F}$, exceeds the threshold
$B_{\mathcal G}^\ast$, so a minimum dropped violation is required before the LATT interval is the shorter
one. The gap between the two point estimators $\hat D=\hat\theta_{\mathcal G}-\hat\theta_{\hat S}$, testable
with the carved standard error of Theorem~\ref{thm:carve} applied to the contrast
$(\ell_{\mathcal G}-\ell_S)'\hat\beta_{\mathrm{post}}$, consistently estimates that dropped violation net
of the composition gap in the informative limit, when the retained set's own residual violation is
negligible.
\item[(iii)] (Sensitivity under heterogeneity.) The two intervals target different objects when the
retained and discarded cohorts differ in effect, so the ranking can misfire on efficiency but never on
coverage. The one quantity the data cannot identify is the composition gap
$\Gamma=\theta_{\hat S}-\theta_{\mathcal G}$. We report it as a breakdown $\Gamma^\ast$ in the same idiom as
$M^\ast$, the amount by which the discarded cohorts' effects must differ from the retained cohorts' to
overturn the reading, with its explicit form $\Gamma^\ast=(\Delta\mathrm{Var}-D^2)/(2D)$ derived in
Appendix~\ref{app:mse}. This is a signed threshold defined for $D\neq0$, with dominance requiring
$\Gamma>\Gamma^\ast$ when $D>0$ and $\Gamma<\Gamma^\ast$ when $D<0$. At $D=0$ the advantage is
$-\Delta\mathrm{Var}\le0$, so no finite breakdown exists, and the plug-in $\hat\Gamma^\ast$ is unstable when
$\hat D$ is near zero.
\end{enumerate}
\end{proposition}

\begin{proof}
See Appendix~\ref{app:proofs}.
\end{proof}

Proposition~\ref{prop:switch} turns the informal ``report the LATT when pre-trends are informative'' into a
rule the data can run, and separates what they decide from what they cannot. The width comparison and $\hat
D$ are estimable, while the unidentified composition gap is reported as the breakdown $\Gamma^\ast$, read as
$M^\ast$ is for the identifying restriction. The rule never suppresses an interval, so it cannot cost
coverage.

Assembled in one place, the procedure is given in Algorithm~\ref{alg:procedure}.

\begin{algorithm}[H]
\caption{The credible-subpopulation procedure}
\label{alg:procedure}
\begin{enumerate}\setlength\itemsep{2pt}
\item Estimate cohort-level (group-time) event studies with any heterogeneity-robust estimator
\citep{callaway2021, sun2021, borusyak2024, dechaisemartin2020}.
\item Select the credible cohorts, from institutional knowledge or a covariate split (ex ante), or
by screening the pre-trends for flatness (data-driven).
\item Estimate the LATT by reweighting the retained cohorts' post-treatment effects.
\item Report honest inference, the level-bound sensitivity interval on the retained set, carved for
the selection when it is data-driven (Theorem~\ref{thm:carve}), together with the breakdown violation
$M^\ast$.
\item Report the trade, both honest intervals (LATT and ATT), preferring whichever is shorter by the
width comparison (Theorem~\ref{thm:width}), with the breakdown value $\Gamma^\ast$ for the unidentified
composition gap (Proposition~\ref{prop:switch}).
\end{enumerate}
\end{algorithm}

\section{Simulation study}\label{sec:sim}

\subsection{Design}\label{sec:design}

We study the method in two environments. In the reduced-form model we draw event-study
coefficients directly from their asymptotic distribution, $\hat\beta\sim\mathcal{N}(\tau+\delta,\Sigma)$
with $\Sigma$ known, which isolates the identification and inference mechanics in the setting where the
underlying theory is exact. In the panel model we generate individual outcomes for a staggered
design and estimate both the coefficients and their covariance from the data, which confirms that the
conclusions survive realistic estimation. Unless stated otherwise, results are averages over eight to
twenty thousand replications.

The cohort structure follows a common template. Some cohorts are clean, with $\delta_g=0$, and the
remainder confounded, with a nonzero differential trend, and Table~\ref{tab:dgp} gives the count and
clean-confounded split used in each exercise. Every cohort is given the same true treatment effect,
normalized to one. This normalization is deliberate, making
the true ATT and the true LATT both equal to one, so that any departure of an estimator from one is
attributable to a violation of parallel trends rather than to treatment-effect heterogeneity. It
therefore isolates the identification problem that the paper is about.

Throughout, each cohort carries a four-period pre-treatment and four-period post-treatment event study
with the reference period normalized to zero. A confounded cohort carries a level violation, its
treated-comparison difference shifting by $V>0$ in the post-period and by $\phi V$ in the pre-period, for
a scalar $\phi\ge 0$ we call the informativeness parameter and that is the master axis of the study. When
$\phi=0$ the confound is flat before treatment and no pre-trend rule can detect it, and as $\phi$ grows
the pre-period shift announces it and it becomes detectable. Because $V$ is the post-treatment violation it
is directly comparable to the level bound $M$ of $\Delta^{\mathrm{Level}}(M)$. In the controlled coverage
exercise below the aggregate being tested carries this violation exactly, which makes the boundary
statement ``covers if and only if $M\ge V$'' checkable. More generally, coverage requires only that $M$
dominate the selected aggregate's violation, which for a set mixing clean and confounded cohorts is weakly
below the largest single-cohort $V$, so $M\ge V$ is sufficient and is necessary only in the worst case of an
all-confounded set. This also makes $\phi$ the operational version of the informativeness of pre-trends
discussed in Section~\ref{sec:approach}.

The selection rule is the flatness screen, posed in the same currency as the honest inference, so cohort
$g$ is retained if its estimated pre-trend is close to flat, $\max_e\lvert\hat\beta_{g,\mathrm{pre}}(e)\rvert\le c$.
The estimator of the ATT averages the post-treatment effect over all cohorts, as a heterogeneity-robust
estimator would, and the LATT averages it over the selected cohorts only. Because the screen reads the selected
set off the same noisy pre-trends, the selection is genuinely data-dependent
in the sense of Section~\ref{sec:approach}, and the pre-test bias of Proposition~\ref{prop:pretest} is
operative.

\begin{table}[t]
\centering
\caption{Simulation parameters. Every exercise uses a four-period pre and four-period post event study with
reference period zero, equal cohort weights, and true effect one for every cohort. The reduced-form model
draws coefficients directly, $\hat\beta\sim\mathcal{N}(\tau+\delta,\Sigma)$ with within-cohort AR(1)
covariance $\Sigma_{ij}=s^2\rho^{|i-j|}$. The panel model draws individual untreated outcomes
$Y_{it}(0)=\delta_g(e)+\epsilon_{it}$, $\epsilon\sim\mathcal{N}(0,\sigma_\epsilon^2)$, for $500$ treated
units per cohort and $1500$ shared controls over $T=14$ periods, and estimates the coefficients and their
covariance from the data.}
\label{tab:dgp}
\begin{tabular}{lcccc}
\toprule
& Estimand & Scope & Honest cov.\ & Panel \\
& (Table~\ref{tab:tier1}) & (Fig.~\ref{fig:scope}) & (Table~\ref{tab:flci}, Fig.~\ref{fig:layer2}) & (Sec.~\ref{sec:calib}) \\
\midrule
Cohorts (clean/confounded) & $6$ ($4/2$) & $12$ ($6/6$) &, & $6$ ($3/3$) \\
Post-treatment violation $V$ & $0.8$ & $0.6$ & $0.3,\ 0.6$ & $0.6$ \\
Informativeness $\phi$ & $1.0$ & $[0,1.5]$ &, & $[0,1.5]$ \\
Screen threshold $c$ & $0.40$ & $0.40$ &, & $0.40$ \\
Noise & $s{=}0.03\text{--}0.40$ & $\sigma{=}0.20$ & $\sigma_{\mathrm{agg}}{=}0.10$ & $\sigma_\epsilon{=}0.5$ \\
Pre/post correlation $\rho$ & $0.5$ & $0.5$ &, & estimated \\
Replications & $20{,}000$ & $8{,}000$ & $4{,}000$ & $2000\text{--}3000$ \\
\bottomrule
\end{tabular}
\end{table}

\subsection{Recovering a clean effect where the ATT is biased}

We first illustrate Proposition~\ref{prop:id} in the simplest configuration, with a subset of cohorts
violating parallel trends in the post-period. Table~\ref{tab:tier1} reports the result. The estimator
of the ATT is biased by $+0.267$, since it averages over all cohorts, so the offending cohorts'
differential trends enter it directly. The LATT lands at $1.000$, on the truth, because the selection
removes those cohorts from the average and, by Proposition~\ref{prop:id}, their violations then do not
enter the estimand at all.

The lower panel of the table sweeps the sampling noise $s$.
As precision improves the LATT's bias falls monotonically from $+0.013$ to zero, and this residual is the
selection-induced bias of Proposition~\ref{prop:pretest}, which vanishes as the selection rule becomes reliable.
Two sources could in principle generate it, a noisy screen that admits genuinely confounded cohorts, which
would persist even under independent selection, and the correlation between a cohort's pre- and
post-treatment sampling errors, which distorts the retained cohorts' estimates only under same-sample
selection. The final column isolates the two by selecting on an independent pre-trend draw. It tracks the
full-sample column almost exactly, so the residual is the first source, admission of confounded cohorts,
and the second is negligible here because the two-sided flatness screen conditions the retained estimates
symmetrically. The persistent component is thus a residual identification bias, present even under
independent selection and addressed by the sensitivity widening of the next section, rather than the strict
post-selection bias that carving removes. The ATT estimator's bias, by contrast, is frozen at $+0.267$ regardless of precision. The
LATT faces a noise problem, which more data cures, whereas the ATT estimator faces a wrong-target problem,
which no amount of data cures.

\begin{table}[t]
\centering
\caption{The estimand demonstration (reduced-form model). True effect $=1$ for all cohorts. In the lower
panel LATT (full) selects on the same pre-trends used for estimation, while LATT (split) selects on an
independent draw. Their near-equality shows the residual bias is admission of confounded cohorts, not
same-sample conditioning.}
\label{tab:tier1}
\begin{tabular}{lccc}
\toprule
& ATT & LATT & LATT \\
& & (full) & (split) \\
\midrule
Point estimate & $1.267$ & $1.000$ &, \\
Bias vs.\ truth & $+0.267$ & $+0.000$ &, \\
\midrule
\multicolumn{4}{l}{Residual LATT bias as sampling noise $s$ falls}\\
$s=0.40$ & $+0.265$ & $+0.013$ & $+0.014$ \\
$s=0.20$ & $+0.267$ & $-0.001$ & $-0.000$ \\
$s=0.12$ & $+0.267$ & $-0.000$ & $-0.000$ \\
$s=0.03$ & $+0.267$ & $+0.000$ & $+0.000$ \\
\bottomrule
\end{tabular}
\end{table}

\subsection{The scope condition}

Figure~\ref{fig:scope} traces the method across the informativeness axis $\phi$ described in
Section~\ref{sec:design}.

The upper-left panel plots the bias of the two point estimators against the truth. The bias of the ATT
estimator is a flat line at $+0.30$, since it never consults pre-trends, so its performance cannot depend
on how informative they are. The LATT's bias begins on top of that line at $\phi=0$ and falls to
zero as pre-trends become informative. The vertical gap between the two curves is therefore exactly the
method's advantage, and reading it off from left to right gives the scope condition, with the advantage
nil when pre-trends carry no information about post-treatment violations and maximal when they carry a
great deal.

The two curves coincide at the left edge because selecting on uninformative pre-trends buys nothing, as
Section~\ref{sec:approach} argues.

The upper-right panel explains why the LATT's bias falls, by plotting the identification gap of
the selected set, the average post-treatment violation among the cohorts that survive selection. It
tracks the LATT's bias almost exactly, falling from $0.30$ to essentially zero. This confirms that the
LATT's remaining bias is not an artifact of the estimator but the residual violation that
selection has failed to remove. The lower-right panel gives the mechanism behind both, the average
number of confounded cohorts that slip through the screen falling from $5.09$ of six when the confound is
invisible to essentially zero when it is plainly visible.

The lower-left panel turns to inference and reports the coverage of a point-based confidence interval
for the causal target, computed both in the naive way and by sample-splitting. Both collapse at low
informativeness and recover only at high informativeness. The comparison is diagnostic, since
sample-splitting removes any contamination from data-dependent selection by construction, so the fact
that it collapses
just as badly establishes that the failure is not selection distortion but the identification gap
of the upper-right panel, since the surviving cohorts genuinely still violate parallel trends. This is what
motivates the sensitivity bounds of the next subsection, since no purely inferential fix can repair a
violation of the identifying assumption.

(At the right edge the naive interval slightly outperforms the split one. This is because
splitting discards half the data, and by that point selection is nearly deterministic, so there is no
contamination left to correct.)

\begin{figure}[t]
\centering
\includegraphics[width=\textwidth]{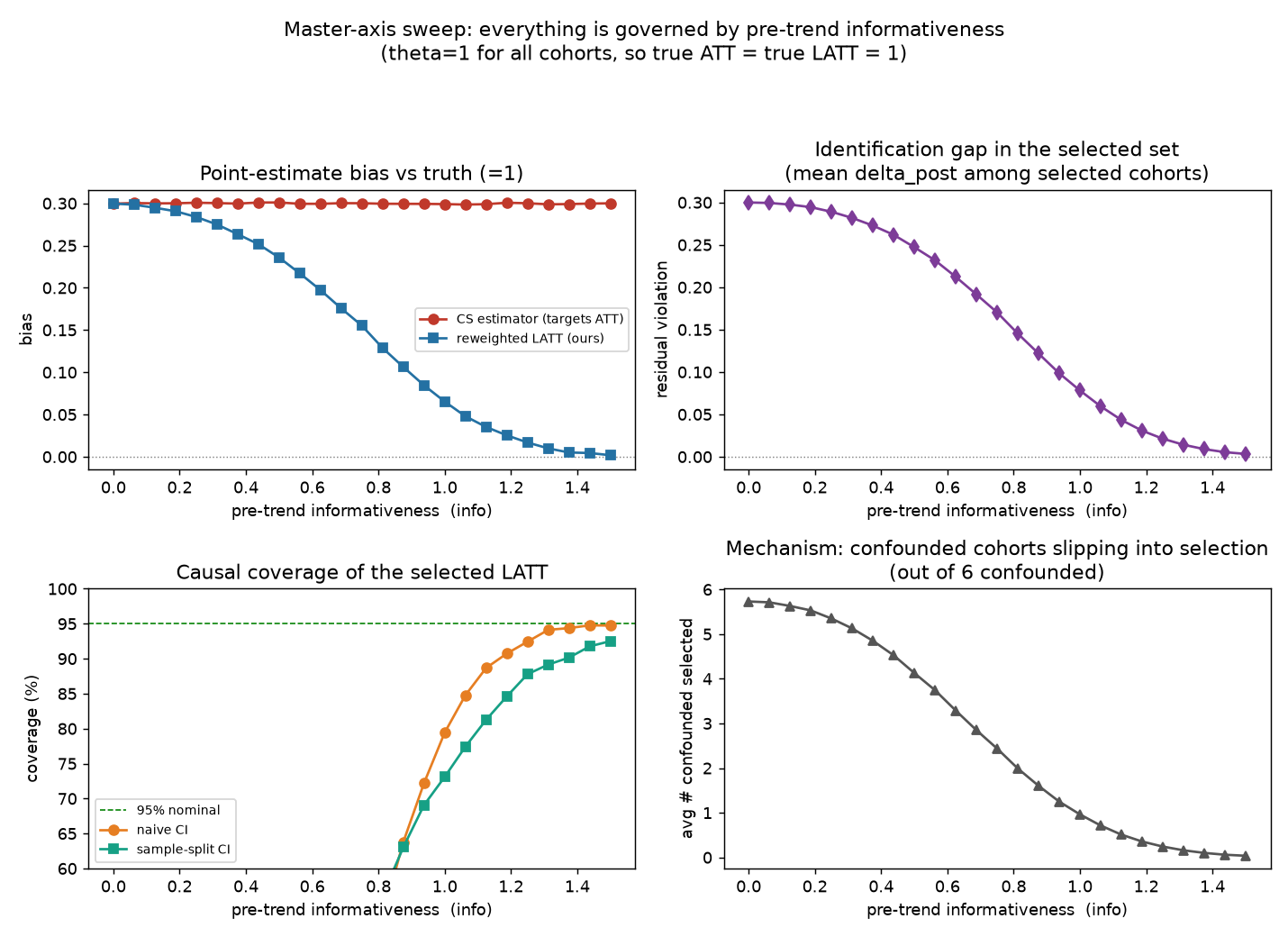}
\caption{Scope condition. Everything is governed by pre-trend informativeness. The ATT estimator's
bias is constant, and the LATT's advantage is the gap between the lines. Causal-coverage collapse at low
informativeness is an identification gap (surviving cohorts still violate parallel trends), not
selection distortion, since sample-splitting collapses just as badly.}
\label{fig:scope}
\end{figure}

\subsection{Honest inference restores coverage}\label{sec:flci}

We now validate the level-bound FLCI of Section~\ref{sec:approach} applied to the selected-cohort
aggregate, using the level violation of Section~\ref{sec:design} whose size $V$ is directly comparable
to the assumed bound $M$.

Table~\ref{tab:flci} is the validity check. The FLCI is honest exactly on the range it claims, nominal at
$V=M$, conservative when $M>V$, collapsing when $M<V$, failing visibly outside its bound rather than
silently. The naive point interval, by contrast, collapses to near-zero coverage the instant any violation
is present, since it is centered on a biased estimate with a width reflecting sampling noise alone.

Figure~\ref{fig:layer2} repeats the exercise across a continuum of violation sizes and adds the practical
output. In the left panel the point interval falls away rapidly as the violation grows, while each FLCI
holds its nominal level until the true violation reaches its own assumed bound, marked by the vertical
dotted lines, and declines thereafter. The middle panel records the price. Interval half-width is
governed by the assumed $M$ rather than by the realized violation, so honesty about a wider class of
violations is paid for in width, immediately and proportionately.

The right panel is what a practitioner would actually report. Holding the data fixed, it traces the
confidence band as the assumed bound is relaxed and identifies the breakdown value $M^\ast$. In this
simulation, where the truth is known, $M^\ast$ is the bound at which the band first admits the true effect.
In an application, where it is not, the operative breakdown is the smallest $M$ at which the band no longer
supports the stated conclusion, for instance first includes zero, and the two need not coincide. Because $M$
is in the outcome's own units, this quantity answers the
question a reader of a DiD paper wants answered directly, how large a post-treatment parallel-trends
violation one must be willing to countenance before the stated conclusion no longer follows. And because
the bound is fixed by the researcher rather than read off the pre-trends, the exercise remains
informative even when pre-trends are uninformative, in direct contrast to the relative-magnitudes
restriction, whose identified set would shrink toward zero in exactly that case, leaving the truth able to
fall outside the maintained class.

\begin{table}[t]
\centering
\caption{FLCI coverage of the causal target (reduced-form model), for an aggregate whose violation equals
$V$, valid iff $M\ge V$.}
\label{tab:flci}
\begin{tabular}{cccc}
\toprule
True violation $V$ & Assumed $M$ & FLCI coverage & Point-estimate coverage \\
\midrule
$0.30$ & $0.30\ (=V)$ & $95.0\%$ & $15.9\%$ \\
$0.30$ & $0.60\ (>V)$ & $100\%$ & $14.9\%$ \\
$0.60$ & $0.30\ (<V)$ & $9.3\%$ & $0.0\%$ \\
$0.60$ & $0.60\ (=V)$ & $95.0\%$ & $0.0\%$ \\
\bottomrule
\end{tabular}
\end{table}

\begin{figure}[t]
\centering
\includegraphics[width=\textwidth]{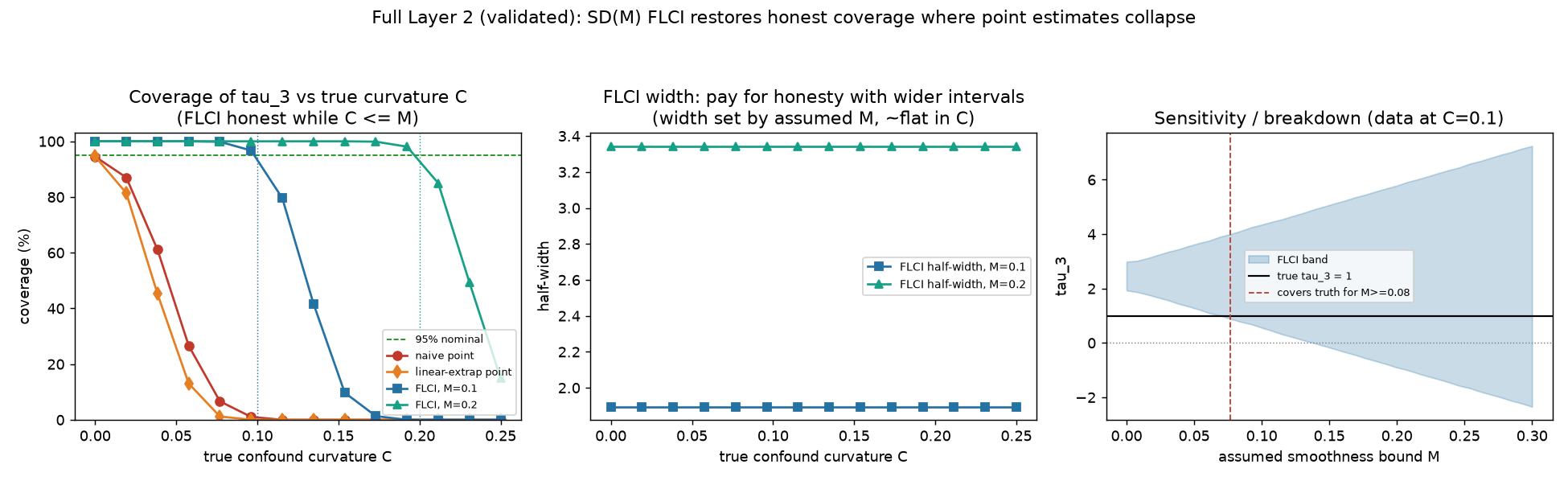}
\caption{Sensitivity bounds restore honest coverage. The point estimate collapses as the violation grows.
The FLCI is honest exactly while $M\ge V$. Interval width is the price of honesty, set by $M$. The right
panel is the sensitivity/breakdown curve, with $M^\ast$ in outcome units.}
\label{fig:layer2}
\end{figure}

\subsection{Width dominance and the shorter honest interval}\label{sec:width}

Having confirmed that the honest interval covers, we return to the decision of Section~\ref{sec:dominance}.
when is it actually shorter than the ATT's? Theorem~\ref{thm:width} answered this analytically, with the
explicit equal-weights threshold $V>1.59\,\sigma$ at $K=2$, $k=1$. We confirm it directly and check the point
that makes the comparison meaningful. The width gain is between two valid intervals, not an artifact
of undercoverage.

The design is the informative limit of Section~\ref{sec:dominance}. Each of $K$ equally weighted cohorts
contributes an independent post-treatment aggregate $\hat\beta_g\sim\mathcal N(\tau_g+V_g,\sigma^2)$, with
$\tau_g=1$ so that the ATT and the LATT are both one and any width difference is attributable to the
violation alone. The $k$ clean cohorts ($V_g=0$) are retained and the $K-k$ confounded ones ($V_g=V$) are
dropped, ex ante, so both intervals are ordinary level-bound FLCIs \citep{armstrong2018}. The LATT's, on the
flat retained set, carries bound $B_S=0$ and standard error $\sigma/\sqrt k$. The ATT's must bound the
full-set violation, $B_{\mathcal G}=fV$ with $f=(K-k)/K$, at standard error $\sigma/\sqrt K$.

Figure~\ref{fig:width} reports the result at $K=2$, $k=1$. Panel (a) plots the two honest half-widths. The
LATT's is flat in $V$, since its retained set is clean, while the ATT's rises as it widens to cover the
growing violation, the two crossing at $V^\ast=1.59\,\sigma$ exactly as Section~\ref{sec:dominance} predicts.
Panel (b) is the check that matters. Across the whole range both honest intervals cover the target at $95\%$,
so past $V^\ast$ the LATT interval is strictly shorter at equal, nominal coverage. The naive ATT
interval, which ignores the violation, is narrower still but collapses in coverage, from $95\%$ to below
$60\%$, and is the interval honest inference is meant to replace. Table~\ref{tab:width} records the
crossing. At $V=2\sigma$ the LATT half-width $1.96\,\sigma$ undercuts the ATT's $2.16\,\sigma$ with both
covering, while the naive interval covers only $71\%$. Panel (c) traces $V^\ast$ across designs. In each
design the width threshold exceeds the corresponding mean-squared-error threshold,
$V/\sigma>\sqrt{1/k-1/K}\,/f$ with $f=(K-k)/K$, which equals $\sqrt2$ at $K=2$, $k=1$, since the
fixed-length interval charges for the standard-error inflation of dropping cohorts and not only for their
variance.

\begin{table}[t]
\centering
\caption{Width dominance ($K=2$, $k=1$, $\sigma=1$, $40{,}000$ replications). Both honest intervals cover the
target at the nominal $95\%$. Past $V^\ast=1.59$ the LATT half-width $h_S$ undercuts the ATT's
$h_{\mathcal G}$ at equal coverage, while the naive ATT interval is narrower only by undercovering.}
\label{tab:width}
\begin{tabular}{cccccc}
\toprule
$V/\sigma$ & $h_S$ & $h_{\mathcal G}$ & LATT cov.\ & ATT cov.\ & naive cov.\ \\
\midrule
$0.80$ & $1.96$ & $1.58$ & $94.9\%$ & $94.9\%$ & $91.2\%$ \\
$1.59$ & $1.96$ & $1.96$ & $94.9\%$ & $95.0\%$ & $79.6\%$ \\
$2.00$ & $1.96$ & $2.16$ & $94.9\%$ & $95.0\%$ & $70.9\%$ \\
$2.50$ & $1.96$ & $2.41$ & $94.9\%$ & $95.0\%$ & $57.8\%$ \\
\bottomrule
\end{tabular}
\end{table}

\begin{figure}[t]
\centering
\includegraphics[width=\textwidth]{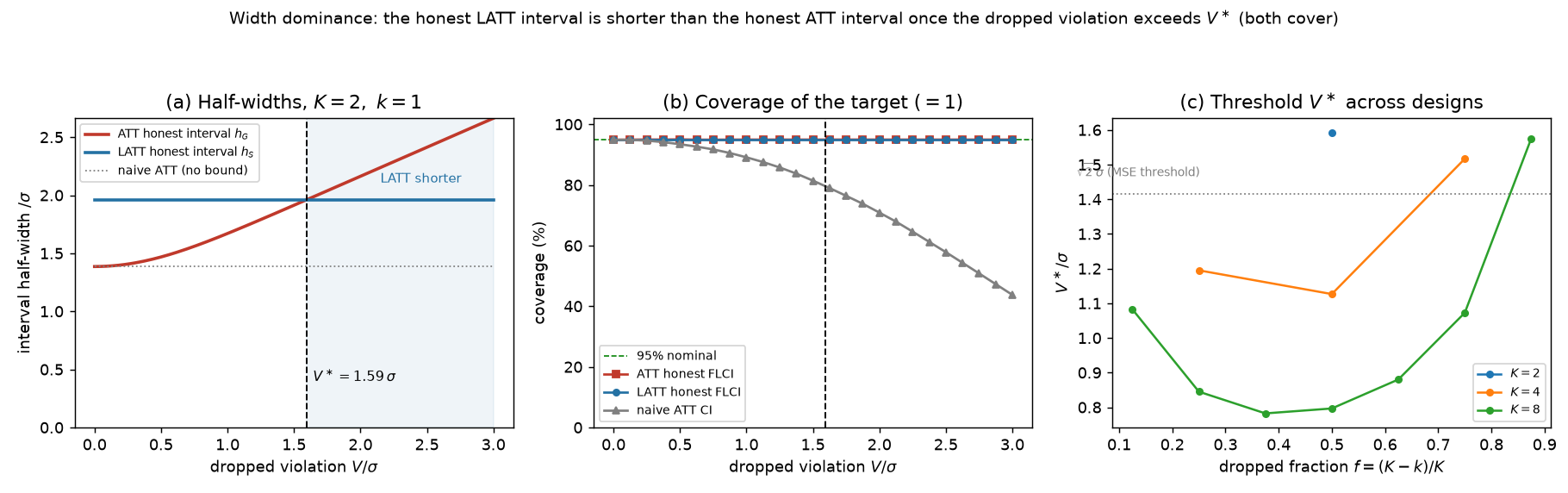}
\caption{Width dominance (Theorem~\ref{thm:width}). (a) The honest LATT and ATT half-widths cross at
$V^\ast=1.59\,\sigma$. Beyond it the LATT interval is shorter. (b) Both honest intervals cover the target at
$95\%$ throughout, so the width gain is not bought by undercoverage, while the naive ATT interval collapses.
(c) The threshold $V^\ast$ across designs, above the design-specific mean-squared-error threshold
($\sqrt2\,\sigma$ at $K=2$, $k=1$).}
\label{fig:width}
\end{figure}

\subsection{Selection, calibration, and panel confirmation}\label{sec:calib}

Two questions remain for the data-driven regime, whether choosing the selected set from the same data
distorts the sensitivity interval, and how to choose $M$ when the selected set, and hence its residual
violation, is random. On the first, the cost of choosing the set from the same data is concentrated near the threshold and falls
on the variance rather than on gross coverage. The $M$ at which coverage crosses the nominal level is nearly
the same for the full-data and sample-splitting procedures ($M\approx0.048$ and $0.056$), so where cohorts
are well separated from $c$ selection is effectively deterministic, naive inference suffices, and splitting
buys little, consistent with the pointwise reading of \citet{leeb2005}. That near-equality of crossing
points understates the cost, however, because selection distorts the standard error rather than the
crossing bound, the local-to-threshold regime the carving of Theorem~\ref{thm:carve} is built for,
which the panel model below makes precise. On the second, the residual violation of the selected aggregate is
not a constant but a spread, multi-modal random variable across replications, running from zero to about
$0.06$ with a thin tail to $0.11$, each mode a different number of borderline cohorts surviving the
screen, which is the content of Remark~\ref{rem:calib}. A fixed $M$ must dominate it, and setting $M$ at
the mean residual violation ($0.043$) undercovers, while nominal coverage is reached near the
ninety-fifth percentile ($0.057$) and well below the maximum ($0.112$). The lesson is not a value at which
to set $M$, the quantile that would deliver a target coverage depends on the unidentified violation
distribution, but a direction in which to read the breakdown curve. A bound drawn from the typical selected
set falls at the undercovering end, so the curve should be read conservatively, with $M^\ast$ weighed against
a violation one is willing to countenance rather than against the set's own average.

We next report the carved interval itself, the object Theorem~\ref{thm:carve} constructs, across the
randomization scale $\gamma$. In a selection-isolation design where every cohort is clean in the
post-period, so the causal target is one for any selected set and the identification widening is held out,
cohorts sit near the flatness threshold and the pre- and post-period estimation errors are correlated, so
the naive interval on the selected aggregate undercovers. Table~\ref{tab:carve} traces the sweep. The
carved interval covers at the nominal level for every $\gamma$, where the naive interval undercovers,
confirming the exact conditional coverage of Theorem~\ref{thm:carve} on the selected set. Its length is
summarized by the median rather than the mean, because the carved interval has infinite expected length
(Proposition~\ref{prop:nolength}): the mean is dominated by the rare near-boundary replications where the
truncation collapses, a heavy tail visible in the ninety-fifth percentile, while the median is stable.
On the median the carved interval is comparable to a proper $\sqrt2$ sample-split, not uniformly shorter,
and its median length grows with $\gamma$ as the noised screen shrinks the retained set and the
empty-selection event, on which the procedure reports the ATT interval, becomes non-negligible. The carved
interval's advantage over splitting is thus not width but that it uses the whole sample deterministically,
with the randomization $\gamma$ trading the length tail against the empty-selection rate; the useful
regime is a small $\gamma$.

\begin{table}[t]
\centering
\caption{The carved interval across the randomization scale $\gamma$ (reduced-form model,
selection-isolation design, causal target $=1$). Coverage is over non-empty selections; length is the
median, since the carved interval has infinite expected length (Proposition~\ref{prop:nolength}). The
carved interval covers where the naive interval undercovers, at a median length comparable to a proper
$\sqrt2$ sample-split.}
\begin{tabular}{cccccccc}
\hline
 & \multicolumn{3}{c}{coverage (\%)} & \multicolumn{3}{c}{length} & \\
$\gamma$ & carved & naive & split & med.\ carved & med.\ split & p90 carved & empty (\%) \\
\hline
$0.0$ & $95.4$ & $90.6$ & $94.6$ & $0.545$ & $0.444$ & $2.47$ & $0.0$ \\
$0.5$ & $95.8$ & $92.2$ & $95.1$ & $0.555$ & $0.444$ & $2.24$ & $0.0$ \\
$2.0$ & $95.1$ & $94.1$ & $95.2$ & $0.621$ & $0.512$ & $2.52$ & $0.7$ \\
$8.0$ & $95.3$ & $94.8$ & $94.5$ & $0.855$ & $0.627$ & $3.17$ & $10.6$ \\
\hline
\end{tabular}
\label{tab:carve}
\end{table}

Finally, we repeat the estimand demonstration on generated panel data, where individual outcomes are
simulated for a staggered design with a level confound and idiosyncratic noise, cohort event-study
coefficients are estimated from genuine difference-in-differences contrasts against a never-treated
group, and their covariance is estimated rather than supplied, which is the empirical counterpart of
Remark~\ref{rem:sigmahat}. The conclusions are unchanged, since
across the informativeness axis the ATT estimator's bias stays flat near $+0.30$ while the LATT's falls
toward zero, reproducing Figure~\ref{fig:scope} with estimated quantities, and the estimated standard
error matches its Monte Carlo counterpart where selection is effectively deterministic. One discrepancy
is informative, and it is the payoff promised above. At intermediate informativeness, where selection is
stochastic, the local-to-threshold regime the carved interval targets, the estimated standard error
understates its Monte Carlo counterpart ($0.016$ against $0.024$ at $\phi=0.75$), the missing
component being the variance induced by the randomness of selection, which a standard error conditional
on the realized set does not capture. This confirms, from the panel side, that conditioning on a
data-dependent selection understates uncertainty, and why the honest object to report is the sensitivity
interval of Section~\ref{sec:flci} rather than a conventional confidence interval.

\section{Empirical application}\label{sec:app}

We illustrate the method with the local effect of the shale boom on house prices, a staggered design
with many counties per cohort. The treatment date is the year a county's oil and gas production first
rises sharply, from county production records \citep{ers2015}, which sorts the boom counties into
onset cohorts, and the outcome is the log county house-price index \citep{bogin2019}, over 1998--2019. We
estimate \citet{callaway2021} group-time effects against never-treated counties, those with negligible
production, screen cohorts on the flatness of their estimated pre-trend, and report the level bound of
Section~\ref{sec:approach} for the selected aggregate.

Onset is gradual and partly endogenous, which is what makes the setting instructive. A county's
production ramps up over several years, and where it does so during a regional housing upswing its
prices are already climbing before the boom is dated. Figure~\ref{fig:frack}(a) shows the result. The
early cohorts have flat pre-trends and are retained, while the later cohorts, whose onset coincides
with the mid-2000s house-price run-up, carry pre-treatment trends several times larger and are dropped.

The two aggregates part ways (Figure~\ref{fig:frack}b). Pooled across all cohorts, the event study
climbs through the onset date to a house-price effect of $+6.7$ log points ($t=8$), the kind of estimate
a practitioner would report as a clear positive effect of the boom. But the pooled path is already
trending before treatment, and the effect is inherited from the dropped cohorts. Restricted to the
credible subpopulation, the three cohorts whose pre-trends are essentially flat (the screen at
$c\approx0.03$), the event study is flat both before and after onset, and the LATT is $-0.2$ log points,
indistinguishable from zero. The gap between the two, $+6.9$ log points, is itself sharply
estimated ($t=5$). Honest inference confirms the reading rather than resting on the point estimate.
Applying the level bound of Section~\ref{sec:approach} to the credible aggregate at the illustrative value $M$
equal to the screen tolerance $c\approx0.03$, one point on the breakdown curve, the LATT's fixed-length
interval runs from $-5.9$ to $+5.5$ log points, which contains zero and excludes the pooled $+6.7$. The
choice $M=c$ is a benchmark rather than a consequence of the screen, since $c$ bounds the pre-treatment
coefficients while $M$ bounds the unobserved post-treatment violation, so the interval is honest only under
the maintained level bound at that $M$. The breakdown value is $M^\ast\approx4.2$ log points, so reconciling the credible near-zero
reading with the pooled effect would require countenancing a post-treatment violation about $1.4$ times the
screen tolerance and larger than any retained cohort's own pre-trend. Panel (c) traces
the estimate along the credibility path of Section~\ref{sec:select}. It stays flat and close to zero
across the whole range of genuinely flat pre-trends, cohorts with $\max_e|\hat\beta_{g,\mathrm{pre}}(e)|$
below about $0.03$, and climbs toward the pooled $+6.7$ only once the threshold is loosened enough to
readmit cohorts whose own pre-trends, around $0.07$--$0.08$, are as large as the effect in question.
Reading the LATT at that lax end simply rebuilds the pooled estimate out of its least credible cohorts,
and the positive figure survives only by abandoning the credibility standard, not by tolerating a
residual violation of the credible set.

Because the credible set is read from the same estimated pre-trends, we run the carved procedure of
Theorem~\ref{thm:carve} on the data, estimating the covariance of the stacked cohort event-study
coefficients by a cluster bootstrap over counties, the estimated-$\Sigma$ case of
Remark~\ref{rem:sigmahat}. The screen boundary is genuinely uncertain here, which is what makes carving
operative rather than cosmetic. Only the 2003 cohort is retained by a wide margin, and 2007, 2008, and
2011 are dropped by one, while four of the nine cohorts sit within about one bootstrap standard error of
the threshold, 2004 and 2005 retained and 2006 and 2009 dropped. Carrying that selection uncertainty, the
carved interval for the credible LATT at $\gamma=0$ is $[-4.3,+3.3]$ log points, tracking the naive
interval because the retained aggregate stays near zero across the plausible selected sets. Composing it
with the level bound at $M=c$, the fully honest data-driven object of Corollary~\ref{cor:causal}, gives
$[-7.3,+6.3]$ log points, which still contains zero with its upper edge below the pooled $+6.7$; the
near-optimal fixed-length interval $[-5.9,+5.5]$ is its counterpart valid under separation
(Remark~\ref{rem:oracle}). The carved machinery thus confirms that the near-zero credible reading is robust
to the selection uncertainty the pre-trend screen introduces, and not an artifact of conditioning on the
realized selected set.

We checked this against the sensitivity analysis a fixed-target approach would actually run, the HonestDiD
implementation of \citet{rambachan2023} applied to the pooled event study. We use their relative-magnitudes
restriction $\Delta^{RM}(\bar M)$, the ``post bounded by pre'' restriction, which
bounds post-treatment trend changes by $\bar M$ times the largest
pre-treatment one, because it reads robustness directly off the observed pre-trend, the mechanism at issue
here. The honest ATT confidence set is $[+1.0,+13.3]$ log points at $\bar M=1$, wider than the
naive interval $[+5.1,+8.4]$, but still excluding zero, and first admits zero only at a breakdown of
$\bar M^\ast=1.2$. The relative-magnitudes restriction thus reports the $+6.7$
effect as robust, surviving post-treatment violations larger than the entire visible pre-trend. That
robustness is an artifact of aggregation. The pooled pre-trend is nearly flat, its largest pre-treatment
coefficient is under one log point, because the confounded cohorts are a minority whose large,
differently-timed pre-trends wash out in the all-cohort average, even as their own pre-trends reach eleven log
points. Reading only the aggregate, the restriction cannot detect this and certifies the effect, the false
precision that Section~\ref{sec:flci} anticipated for a bound read off an uninformative pre-trend, here on
real data.

The complementary smoothness restriction $\Delta^{SD}(M)$ fails in the opposite direction, for a reason
specific to this confound. Run on the same pooled event study through the same conditional--hybrid procedure,
its honest confidence set admits zero once the post-treatment differential trend is allowed to curve by
$M\approx0.001$, a fifth of the $\approx\!0.005$ curvature the pooled pre-trend itself displays, so it
neither certifies the $+6.7$ effect nor localizes it, but widens to include both zero and the full
effect.\footnote{We report $\Delta^{SD}$ through the same conditional least-favourable hybrid used for
$\Delta^{RM}$, so the two restrictions are compared under one inference procedure. The smoothness class also
admits the near-optimal fixed-length interval of \citet{armstrong2018}, which is tighter but leaves the
comparison unchanged. Near the breakdown the identified set is small and the hybrid is close to optimal.}
The housing bubble is a smooth secular appreciation, and a smoothness restriction cannot separate it from a
treatment effect on the aggregate. The two standard restrictions thus fail in opposite ways, relative
magnitudes falsely precise, smoothness uninformative, and neither recovers the near-zero reading the
cohort-level screen delivers.

The credible-subpopulation screen, reading pre-trends cohort by cohort, drops exactly the bubble cohorts and
returns the near-zero LATT. It is the only one of the three analyses that both shows no effect for the
credible cohorts and locates the pooled positive in the dropped ones, because the information distinguishing
bubble from effect lives in the cross-cohort composition that every aggregate analysis first averages away.
Whether the dropped cohorts' excess is differential trend or genuine heterogeneous effect is the unidentified
composition gap, so the analysis concentrates the pooled positive in cohorts with nonflat pre-trends rather
than proving the full-population effect zero.

The report of Section~\ref{sec:dominance} makes this quantitative. The estimable divergence between the
two point estimators is $\hat D=\hat\theta_{\mathcal G}-\hat\theta_{\hat S}=+6.9$ log points ($t=5.4$,
carved standard error), and it splits as $\hat D=B_{\mathcal F}-\Gamma$ into the dropped cohorts'
differential trend $B_{\mathcal F}$, a pre-trend violation, and the composition gap $\Gamma$, a difference
in true effects, which the data cannot separate. The breakdown value is $\Gamma^\ast\approx-3.3$ log
points, so the credible reading is overturned only if more than about half of the divergence, $3.3$ of the
$6.9$ log points, reflects genuine treatment-effect heterogeneity rather than differential pre-trend.
Reporting both honest intervals with $\Gamma^\ast$ completes the standing rule; the width comparison of the
switching rule is an efficiency question, demonstrated in Section~\ref{sec:width}, and here the two honest
intervals disagree on substance rather than width.

\begin{figure}[t]
\centering
\includegraphics[width=\textwidth]{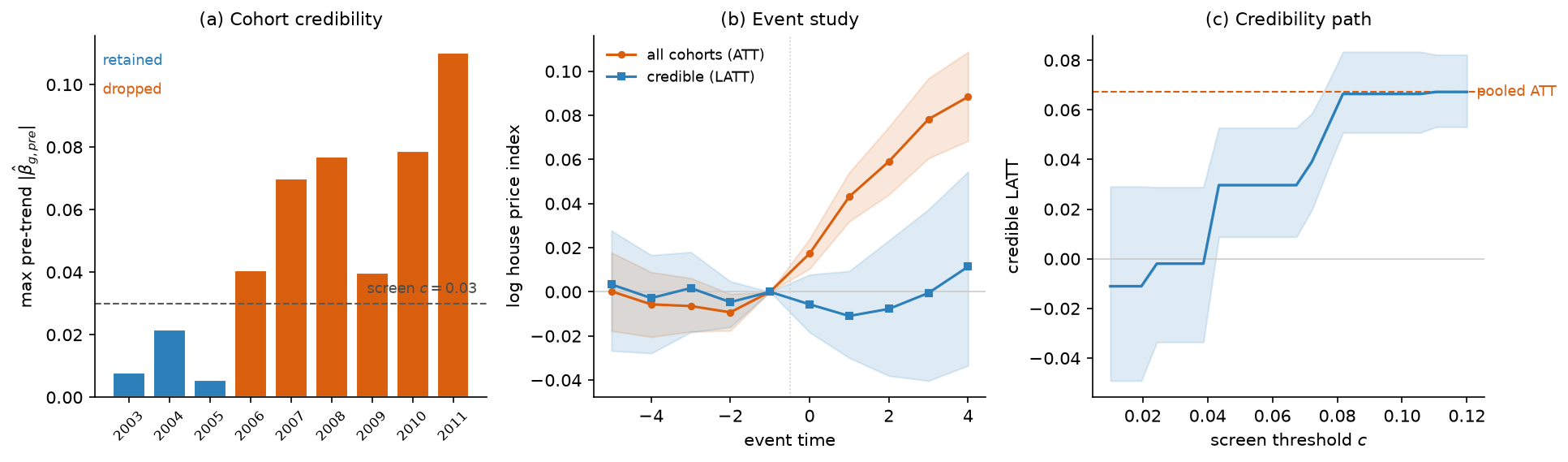}
\caption{The shale boom and county house prices. (a) Cohort credibility map, the maximum pre-treatment
violation by onset cohort with the flatness screen, where the later, trending cohorts are dropped. (b)
Event study, all cohorts (pooled estimate) versus the credible subpopulation (LATT) at the illustrative
threshold $c\approx0.03$, with $95\%$ bands. (c) The credibility path, the LATT as the screen threshold
$c$ is relaxed, near zero across the flat-pre-trend cohorts and rising toward the pooled estimate only as
trending cohorts are readmitted.}
\label{fig:frack}
\end{figure}

The reading is that the pooled effect reflects a pre-existing trend in the counties that adopt late rather
than a credibly causal effect, and a near-zero net effect for the credible subpopulation is what the
offsetting amenity and disamenity channels of local energy development would predict
\citep{muehlenbachs2015}. The method identifies the effect for the credible cohorts alone, and the excess
in the dropped cohorts blends their steep pre-trends with any genuine difference in their treatment
effects, a mixture the design cannot separate and need not, since the credible reading stands on the
retained cohorts. The application thus shows the method working in the direction opposite to
point-identification, not recovering an effect the pooled estimate misses but withdrawing one it spuriously
reports, and locating that withdrawal exactly in the cohorts the screen flags.

\section{Conclusion}\label{sec:conc}

We have argued that when parallel trends fails for some treated cohorts, a productive response is to
change the estimand rather than to defend or bound the ATT. The credible-subpopulation LATT is
point-identified under parallel trends for the selected cohorts alone, a condition that can hold when the
ATT's fails, is estimated by a transparent reweighting of standard group-time effects, and can be
accompanied by honest sensitivity bounds on the residual violation that the pre-trend screen cannot rule
out. This is a specific instance of the general principle, familiar from instrumental
variables and limited overlap, of retreating to a credibly-identified subpopulation.

The method delivers a sharp answer where the ATT estimator is biased and a fixed-target sensitivity
analysis on the aggregate is uninformative, wide when the aggregate pre-trend reveals the violation, and
falsely precise when the offending cohorts wash out of it, at the price of a narrower, subpopulation target
whose composition must be described rather than assumed. Its advantage is contingent on pre-trends being informative about post-treatment violations.
Where they are not, only economic context can carry the assumption. Under data-driven selection the
estimator carries a pre-test bias that vanishes as pre-trends grow informative. The inference we
recommend is honest regardless, combining the post-selection carving of \citet{lee2016} with the
sensitivity bounds of \citet{rambachan2023} into an interval uniformly valid against both the selection and
the residual violation. We also make the scope condition operational. A feasible rule, reported alongside both honest
intervals, signals when narrowing the target lowers risk, with a breakdown value $\Gamma^\ast$ standing in
for the one quantity, the effects of the discarded cohorts, that the data cannot identify.

An appendix characterizes the optimal credible weighting (Proposition~\ref{prop:optimal}). Minimizing
worst-case risk yields a soft-thresholded weighting of cohorts by credibility net of noise, of which the
flatness screen with size weights is a special case, which completes the analogy to the variance-minimizing
overlap weights of \citet{crump2009} that the approach borrows from.

Two directions remain for future work. One is the finite-sample behaviour of the feasible optimal weighting
of Appendix~\ref{sec:optimal}, where the credibility levels $m_g$ are estimated from the pre-trends rather
than known, and the noise this inherits awaits a full account. The other is the coverage of a single
interval selected by a hard-thresholded switching rule. The recommended report of both intervals is honest
whatever the rule decides (Proposition~\ref{prop:switch}), so the question is confined to the coverage of one
interval at the switching boundary.

\clearpage
\bibliographystyle{aer}
\bibliography{refs}

@article{armstrong2018,
  author  = {Armstrong, Timothy B. and Koles{\'a}r, Michal},
  title   = {Optimal Inference in a Class of Regression Models},
  journal = {Econometrica},
  year    = {2018},
  volume  = {86},
  number  = {2},
  pages   = {655--683}
}

@article{borusyak2024,
  author  = {Borusyak, Kirill and Jaravel, Xavier and Spiess, Jann},
  title   = {Revisiting Event-Study Designs: Robust and Efficient Estimation},
  journal = {Review of Economic Studies},
  year    = {2024},
  note    = {Forthcoming}
}

@article{callaway2021,
  author  = {Callaway, Brantly and Sant'Anna, Pedro H. C.},
  title   = {Difference-in-Differences with Multiple Time Periods},
  journal = {Journal of Econometrics},
  year    = {2021},
  volume  = {225},
  number  = {2},
  pages   = {200--230}
}

@article{abadie2010,
  author  = {Abadie, Alberto and Diamond, Alexis and Hainmueller, Jens},
  title   = {Synthetic Control Methods for Comparative Case Studies: Estimating the Effect of California's Tobacco Control Program},
  journal = {Journal of the American Statistical Association},
  year    = {2010},
  volume  = {105},
  number  = {490},
  pages   = {493--505},
}

@article{crump2009,
  author  = {Crump, Richard K. and Hotz, V. Joseph and Imbens, Guido W. and Mitnik, Oscar A.},
  title   = {Dealing with Limited Overlap in Estimation of Average Treatment Effects},
  journal = {Biometrika},
  year    = {2009},
  volume  = {96},
  number  = {1},
  pages   = {187--199}
}

@article{dechaisemartin2020,
  author  = {de Chaisemartin, Cl{\'e}ment and D'Haultf{\oe}uille, Xavier},
  title   = {Two-Way Fixed Effects Estimators with Heterogeneous Treatment Effects},
  journal = {American Economic Review},
  year    = {2020},
  volume  = {110},
  number  = {9},
  pages   = {2964--2996}
}

@article{dechaisemartin2026,
  author  = {de Chaisemartin, Cl{\'e}ment},
  title   = {Using Pre-Trends for Inference in Difference-in-Differences},
  journal = {arXiv preprint arXiv:2607.21312},
  year    = {2026}
}

@article{dechaisemartin2024pretest,
  author  = {de Chaisemartin, Cl{\'e}ment and D'Haultf{\oe}uille, Xavier},
  title   = {Is Inference Conditional on Not Rejecting a Pre-test Less Reliable than Unconditional Inference?},
  journal = {arXiv preprint arXiv:2407.03725},
  year    = {2024}
}

@article{goodmanbacon2021,
  author  = {Goodman-Bacon, Andrew},
  title   = {Difference-in-Differences with Variation in Treatment Timing},
  journal = {Journal of Econometrics},
  year    = {2021},
  volume  = {225},
  number  = {2},
  pages   = {254--277}
}

@article{imbens1994,
  author  = {Imbens, Guido W. and Angrist, Joshua D.},
  title   = {Identification and Estimation of Local Average Treatment Effects},
  journal = {Econometrica},
  year    = {1994},
  volume  = {62},
  number  = {2},
  pages   = {467--475}
}

@article{kahnlang2020,
  author  = {Kahn-Lang, Ariella and Lang, Kevin},
  title   = {The Promise and Pitfalls of Differences-in-Differences: Reflections on
             \emph{16 and Pregnant} and Other Applications},
  journal = {Journal of Business and Economic Statistics},
  year    = {2020},
  volume  = {38},
  number  = {3},
  pages   = {613--620}
}

@article{leeb2005,
  author  = {Leeb, Hannes and P{\"o}tscher, Benedikt M.},
  title   = {Model Selection and Inference: Facts and Fiction},
  journal = {Econometric Theory},
  year    = {2005},
  volume  = {21},
  number  = {1},
  pages   = {21--59}
}

@article{leeb2008,
  author  = {Leeb, Hannes and P{\"o}tscher, Benedikt M.},
  title   = {Sparse Estimators and the Oracle Property, or the Return of Hodges' Estimator},
  journal = {Journal of Econometrics},
  year    = {2008},
  volume  = {142},
  number  = {1},
  pages   = {201--211}
}

@article{kivaranovicleeb2021,
  author  = {Kivaranovic, Danijel and Leeb, Hannes},
  title   = {On the Length of Post-Model-Selection Confidence Intervals Conditional on Polyhedral Constraints},
  journal = {Journal of the American Statistical Association},
  year    = {2021},
  volume  = {116},
  number  = {534},
  pages   = {845--857}
}

@article{li2018,
  author  = {Li, Fan and Morgan, Kari Lock and Zaslavsky, Alan M.},
  title   = {Balancing Covariates via Propensity Score Weighting},
  journal = {Journal of the American Statistical Association},
  year    = {2018},
  volume  = {113},
  number  = {521},
  pages   = {390--400}
}

@article{manski2018,
  author  = {Manski, Charles F. and Pepper, John V.},
  title   = {How Do Right-to-Carry Laws Affect Crime Rates? Coping with Ambiguity
             Using Bounded-Variation Assumptions},
  journal = {Review of Economics and Statistics},
  year    = {2018},
  volume  = {100},
  number  = {2},
  pages   = {232--244}
}

@article{rambachan2023,
  author  = {Rambachan, Ashesh and Roth, Jonathan},
  title   = {A More Credible Approach to Parallel Trends},
  journal = {Review of Economic Studies},
  year    = {2023},
  volume  = {90},
  number  = {5},
  pages   = {2555--2591}
}

@article{kwonroth2024,
  author  = {Kwon, Soonwoo and Roth, Jonathan},
  title   = {(Empirical) Bayes Approaches to Parallel Trends},
  journal = {AEA Papers and Proceedings},
  year    = {2024},
  volume  = {114},
  pages   = {606--609}
}

@article{roth2022pretrends,
  author  = {Roth, Jonathan},
  title   = {Pretest with Caution: Event-Study Estimates after Testing for Parallel Trends},
  journal = {American Economic Review: Insights},
  year    = {2022},
  volume  = {4},
  number  = {3},
  pages   = {305--322}
}

@article{roth2023trending,
  author  = {Roth, Jonathan and Sant'Anna, Pedro H. C. and Bilinski, Alyssa and Poe, John},
  title   = {What's Trending in Difference-in-Differences? A Synthesis of the Recent
             Econometrics Literature},
  journal = {Journal of Econometrics},
  year    = {2023},
  volume  = {235},
  number  = {2},
  pages   = {2218--2244}
}

@article{sun2021,
  author  = {Sun, Liyang and Abraham, Sarah},
  title   = {Estimating Dynamic Treatment Effects in Event Studies with Heterogeneous
             Treatment Effects},
  journal = {Journal of Econometrics},
  year    = {2021},
  volume  = {225},
  number  = {2},
  pages   = {175--199}
}

@article{freyaldenhoven2019,
  author  = {Freyaldenhoven, Simon and Hansen, Christian and Shapiro, Jesse M.},
  title   = {Pre-event Trends in the Panel Event-Study Design},
  journal = {American Economic Review},
  year    = {2019},
  volume  = {109},
  number  = {9},
  pages   = {3307--3338}
}

@article{bilinski2020,
  author  = {Bilinski, Alyssa and Hatfield, Laura A.},
  title   = {Nothing to See Here? Non-inferiority Approaches to Parallel Trends
             and Other Model Assumptions},
  journal = {arXiv preprint arXiv:1805.03273},
  year    = {2020}
}

@misc{ers2015,
  author       = {{U.S. Department of Agriculture, Economic Research Service}},
  title        = {County-level Oil and Gas Production in the United States},
  year         = {2015},
  note         = {Onshore production, lower 48 states, 2000--2011}
}

@article{bogin2019,
  author  = {Bogin, Alexander N. and Doerner, William M. and Larson, William D.},
  title   = {Local House Price Dynamics: New Indices and Stylized Facts},
  journal = {Real Estate Economics},
  year    = {2019},
  volume  = {47},
  number  = {2},
  pages   = {365--398}
}

@article{muehlenbachs2015,
  author  = {Muehlenbachs, Lucija and Spiller, Elisheba and Timmins, Christopher},
  title   = {The Housing Market Impacts of Shale Gas Development},
  journal = {American Economic Review},
  year    = {2015},
  volume  = {105},
  number  = {12},
  pages   = {3633--3659}
}

@article{lee2016,
  author  = {Lee, Jason D. and Sun, Dennis L. and Sun, Yuekai and Taylor, Jonathan E.},
  title   = {Exact Post-Selection Inference, with Application to the Lasso},
  journal = {The Annals of Statistics},
  year    = {2016},
  volume  = {44},
  number  = {3},
  pages   = {907--927},
}

\clearpage
\appendix

\section{Optimal credible weighting}\label{sec:optimal}

The flatness screen and the cohort weights $w_g$ have so far been taken as given, the screen retaining a
cohort outright and the weights fixed at, say, cohort size. Both are choices, and the sensitivity class
that governs the honest inference also pins down their optimal form, in the decision-theoretic sense in
which \citet{crump2009} derive the optimal-overlap subpopulation and \citet{li2018} the overlap weights.
There, variance is minimized where positivity is scarce. Here, worst-case mean squared error is minimized
where credibility is scarce.

Let the estimator be a weighted average $\hat\theta(\lambda)=\sum_g\lambda_g\hat\beta_{g,\mathrm{post}}$
with $\lambda_g\ge0$ and $\sum_g\lambda_g=1$, and, following the overlap-weighting logic that the weights
define the estimand as well as the estimator, let the target be the $\lambda$-weighted causal effect
$\theta(\lambda)=\sum_g\lambda_g\theta_g$, whose interpretation is reported through the weights themselves.
Since $\hat\beta_{g,\mathrm{post}}\sim\mathcal N(\theta_g+V_g,\sigma_g^2)$ independently, the estimator has
identification bias $\sum_g\lambda_g V_g$ and variance $\sum_g\lambda_g^2\sigma_g^2$. The pre-trends enter
through the sensitivity class. A cohort's post-treatment violation is bounded by a credibility level
$m_g\ge0$, $\lvert V_g\rvert\le m_g$, with $m_g$ small for a cohort whose pre-trend is flat and large for
one whose pre-trend is steep. The hard flatness screen is the coarsening $m_g\in\{0,\infty\}$, retaining
the $m_g=0$ cohorts and discarding the rest. The level bound $M$ of Section~\ref{sec:honest} is the common
value $m_g\equiv M$ on the retained set. The worst-case mean squared error over this class is
\[
\mathcal R(\lambda)=\Big(\sum_g\lambda_g m_g\Big)^2+\sum_g\lambda_g^2\sigma_g^2,
\]
the squared worst-case bias plus the variance, and the optimal credible weighting minimizes it over the
simplex.

\begin{proposition}[Optimal credible weighting]\label{prop:optimal}
$\mathcal R$ is strictly convex on the simplex, so it has a unique minimizer $\lambda^\ast$. There are
scalars $\mu$ and $B=\sum_g\lambda_g^\ast m_g$, the worst-case bias at the optimum, such that
\[
\lambda_g^\ast=\frac{1}{\sigma_g^2}\Big[\tfrac{\mu}{2}-B\,m_g\Big]_+,\qquad [x]_+=\max(x,0),
\]
with $\mu$ fixed by $\sum_g\lambda_g^\ast=1$. When $B>0$ there is an endogenous credibility cutoff
$m^\ast=\mu/(2B)$. Cohorts with $m_g\ge m^\ast$ receive zero weight, and among the retained cohorts the
weight is $\lambda_g^\ast\propto(m^\ast-m_g)/\sigma_g^2$, decreasing in the credibility bound $m_g$ and in
the sampling variance $\sigma_g^2$. When $B=0$, as under the benchmark $m_g\equiv0$, the cutoff is vacuous
and $\lambda_g^\ast\propto1/\sigma_g^2$, inverse-variance weights with no thresholding.
\end{proposition}

\begin{proof}
$\sum_g\lambda_g^2\sigma_g^2$ is strictly convex for $\sigma_g^2>0$ and $(\sum_g\lambda_g m_g)^2$ is
convex, so $\mathcal R$ is strictly convex and its minimizer over the compact convex simplex is unique. The
worst-case bias is $\sup_{\lvert V_g\rvert\le m_g}\lvert\sum_g\lambda_g V_g\rvert=\sum_g\lambda_g m_g$
since $\lambda_g\ge0$, giving the displayed $\mathcal R$. Form the Lagrangian
$\mathcal R(\lambda)-\mu(\sum_g\lambda_g-1)-\sum_g\nu_g\lambda_g$ with $\nu_g\ge0$. Stationarity gives
$2Bm_g+2\lambda_g\sigma_g^2-\mu-\nu_g=0$ with $B=\sum_g\lambda_g m_g$. For $\lambda_g>0$ complementary
slackness sets $\nu_g=0$ and $\lambda_g=(\mu/2-Bm_g)/\sigma_g^2$. For $\lambda_g=0$ it requires
$\mu/2-Bm_g\le0$, i.e.\ $m_g\ge\mu/(2B)$. Combining the two cases gives the truncated form and the cutoff
$m^\ast=\mu/(2B)$, with $\mu$ determined by the budget constraint and $B$ by its own definition at the
resulting $\lambda^\ast$.
\end{proof}

Three readings connect the result to the rest of the paper. First, the optimal rule is a soft
flatness screen. It reproduces the hard screen's qualitative content, cohorts whose pre-trends are too
steep are dropped, but with an endogenous threshold set by the bias--variance trade-off rather than a
researcher's $c$, and with the surviving cohorts weighted continuously by how credible and how precise they
are, rather than equally. A marginally less credible but large and precisely estimated cohort is retained
where the hard screen would discard it, because the variance its inclusion saves outweighs the bias it
admits, the same calculus as the dominance condition of Proposition~\ref{prop:generalK}, now internal to
the weighting. Second, the result justifies the paper's default. Among equally credible cohorts,
$m_g\equiv0$, the weights collapse to inverse-variance, $\lambda_g^\ast\propto1/\sigma_g^2$, and when
sampling variance scales inversely with cohort size, $\sigma_g^2\propto1/n_g$, inverse-variance weighting
is size weighting. The fixed $w_g$ of the definition is thus optimal precisely under homoskedastic-per-unit
noise and equal credibility, and the proposition says what to do when either fails. Third, it completes the
analogy to limited overlap. \citet{crump2009} trim and \citet{li2018} weight to minimize variance where
positivity is the binding scarcity. Here credibility is the scarcity, and the optimal response is to weight
cohorts by credibility net of noise, with trimming emerging as the corner solution for cohorts past
$m^\ast$.

The credibility levels $m_g$ are reported, not assumed, in the same currency as everything else. A natural
choice sets $m_g$ to the cohort's own maximal pre-trend, $m_g=\max_{e<0}\lvert\hat\beta_{g,\mathrm{pre}}(e)
\rvert$, so that the weighting is read off the observed pre-trends. This makes the weights data-driven, a
refinement of the selection already treated, and, as there, validity does not rest on it. The honest
interval of Section~\ref{sec:honest} takes its coverage from the researcher-set level bound $M$, while the
estimated $m_g$ only orders and weights the cohorts. Using $m_g$ to weight but $M$ to cover keeps efficiency
and honesty on separate instruments, as the carved inference of Theorem~\ref{thm:carve} keeps
selection and identification on separate instruments. One caveat separates the feasible rule from
Proposition~\ref{prop:optimal}. That optimality is minimax over the class $\lvert V_g\rvert\le m_g$, so it
holds for the feasible weights only if the observed pre-trend $m_g$ bounds the post-treatment violation
$V_g$. Absent such a pre-to-post link, the feasible $m_g$ is a credibility ranking rather than the bound
defining the optimized class, and the weighting is a heuristic that need not inherit the minimax
interpretation. A second caveat concerns coverage. Because the estimated weights make the contrast $\ell_S$
and the causal target random functions of the pre-trends, the discrete-selection guarantee of
Theorem~\ref{thm:carve}, which conditions on a fixed contrast, does not extend to the continuously weighted
interval. The level bound $M$ absorbs the residual identification bias, but the sampling and selection
uncertainty from estimating the weights requires a separate argument, which we leave to future work.
A practitioner wanting the feasible weighting with valid inference today can estimate the weights on an
independent split and invert on the other, at the usual efficiency cost; honest inference that carves the
continuous weighting without that split is the open question.

\subsection{The optimal weighting improves on the hard screen}\label{sec:optsim}

Proposition~\ref{prop:optimal} predicts that the soft, credibility-and-precision weighting lowers risk
relative to the hard flatness screen. We confirm this and measure the gain. The design has $K=12$ cohorts
with homogeneous effects, so every weighting targets the same value and mean squared error is a fair
comparison. Cohorts come in pairs at six violation levels $V_g\in\{0,0.08,\dots,0.40\}$. Within each pair
one member is precise ($n=6400$) and one noisy ($n=800$), with per-cohort post-treatment standard deviation
$\sigma_g\propto1/\sqrt{n_g}$ from a common unit scale, so the pair shares a violation level but not a
variance. A hard screen keeps or drops a pair as a
whole, since both members share the credibility $m_g$, whereas the optimal weighting can separate them. The
design is deliberately constructed to exhibit this precision channel, and a random-DGP average below reports
the typical gain.

We compare three procedures for the common target, the ATT (equal weights on all cohorts), the hard screen
(equal weights on $\{g:m_g\le c\}$ at its best threshold $c$), and the optimal weighting $\lambda^\ast$ of
Proposition~\ref{prop:optimal}, in an oracle version with $m_g=V_g$ known and a feasible version with
$\hat m_g=\max_e\lvert\hat\beta_{g,\mathrm{pre}}(e)\rvert$ estimated from pre-trends at informativeness
$\phi=1.5$. The hard-screen threshold $c$ is chosen to minimize MSE, and the $500$-DGP average draws random
violation levels and sample sizes from the same family.

Table~\ref{tab:prop11} and Figure~\ref{fig:prop11} report the result. The left panel plots MSE against the
hard-screen threshold. The optimal weighting lies below the hard-screen curve at every $c$, including its
minimum. Against the best hard screen, the oracle optimal weighting ($m_g=V_g$) lowers MSE from $0.00104$ to
$0.00039$, by $63\%$. The feasible weighting, estimating $\hat m_g$ from the pre-trends, attains $0.00048$, a
$54\%$ reduction against the same tabulated hard screen. Measured instead against a feasible hard screen that
also estimates $\hat m_g$, not shown in Table~\ref{tab:prop11}, the reduction is $28\%$. Across $500$ random
DGPs the average reduction over the best hard screen is $14\%$. The right panel shows the mechanism. At a common credibility level the optimal weighting loads on the
precise member and down-weights the noisy one, where the hard screen weights them equally, and it retains
marginally less credible cohorts at reduced weight rather than dropping them at a cutoff. The gap between the
oracle and feasible lines is the price of estimating $m_g$ from noisy pre-trends, the open question the
conclusion notes. Even so, the feasible weighting improves on the hard screen.

\begin{table}[t]
\centering
\caption{Optimal weighting versus the hard screen (reduced-form model, homogeneous effects, so all
procedures target the same value). RMSE and MSE for that target. The optimal weighting minimizes worst-case
MSE by Proposition~\ref{prop:optimal}.}
\label{tab:prop11}
\begin{tabular}{lcc}
\toprule
Procedure & RMSE & MSE \\
\midrule
ATT (all cohorts) & $0.200$ & $0.0402$ \\
Hard screen (best $c$) & $0.032$ & $0.00104$ \\
Optimal weighting, feasible ($\hat m_g$) & $0.022$ & $0.00048$ \\
Optimal weighting, oracle ($m_g=V_g$) & $0.020$ & $0.00039$ \\
\bottomrule
\end{tabular}
\end{table}

\begin{figure}[t]
\centering
\includegraphics[width=\textwidth]{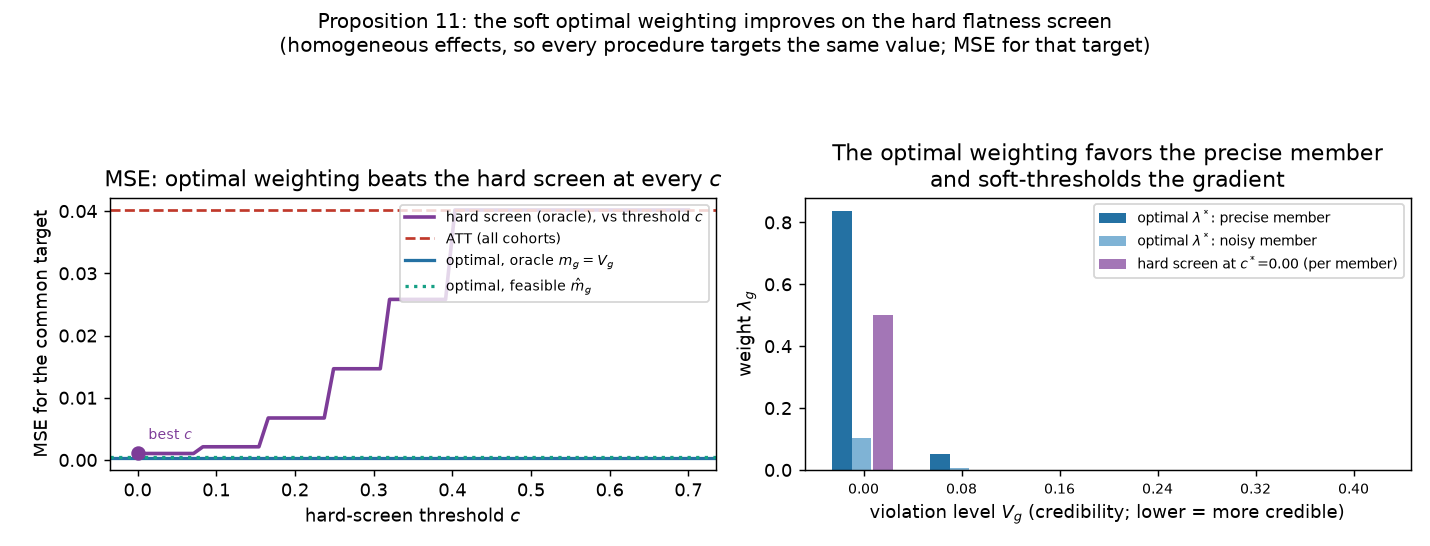}
\caption{Optimal credible weighting (Proposition~\ref{prop:optimal}) versus the hard flatness screen,
homogeneous effects so every procedure targets the same value. Left, MSE against the hard-screen threshold
$c$. The optimal weighting, oracle and feasible, lies below the hard-screen curve at every $c$ and below the
ATT. Right, weights at each violation level, the optimal weighting loads on the precise member of each pair
and soft-thresholds the credibility gradient, where the hard screen weights retained cohorts equally.}
\label{fig:prop11}
\end{figure}

\section{Mean-squared-error dominance}\label{app:mse}

This appendix develops the point-estimator account behind the width comparison of
Section~\ref{sec:dominance}, identifying which three forces set the sign of the advantage and the explicit
thresholds they imply. Changing the estimand is worth doing only when the bias it removes outweighs the variance and
the change of target it costs. Three forces set the sign, the violation the screen removes, the variance
that dropping a cohort costs, and the heterogeneity between the retained subpopulation and the full
one, and the sign turns on their balance, not on the informativeness of the pre-trends, which governs only
how much of a fixed-sign advantage is realized.

We state the comparison for $K$ cohorts, and then read from the two-cohort special case
the one thing the general result leaves implicit, how the advantage scales with the
informativeness of the pre-trends. Let each cohort $g$ carry a weight
$w_g$ (normalized so $\sum_{g}w_g=1$), a causal effect $\theta_g=\mathrm{ATT}(g,\cdot)$,
a post-treatment violation $V_g$ (the aggregate $a(\delta_{g,\mathrm{post}})$, zero for a
clean cohort), and an independent sampling variance $\sigma_g^2$ for its aggregated
post-treatment estimate. The screen retains cohort $g$ with probability $\pi_g$,
independently across cohorts and of the post-treatment estimates. The ATT estimator is
the fixed average $\sum_g w_g\hat\beta_{g,\mathrm{post}}$. The LATT estimator reweights
over the retained set $\hat S$, $\sum_{g\in\hat S}w_g\hat\beta_{g,\mathrm{post}}/W_{\hat S}$
with $W_S=\sum_{g\in S}w_g$.

\begin{proposition}[General-$K$ dominance]\label{prop:generalK}
With the target $\theta_{\mathcal G}=\sum_g w_g\theta_g$, the ATT estimator has mean squared
error $\mathrm{MSE}_{\mathrm{ATT}}=\big(\sum_g w_gV_g\big)^2+\sum_g w_g^2\sigma_g^2$, and the
LATT estimator has
\begin{equation}\label{eq:latt-exact}
\mathrm{MSE}_{\mathrm{LATT}}=\sum_{S\subseteq\mathcal G}
\Big(\textstyle\prod_{g\in S}\pi_g\prod_{g\notin S}(1-\pi_g)\Big)
\Big[\big(\theta_S-\theta_{\mathcal G}+B_S\big)^2+\mathrm{Var}_S\Big],
\end{equation}
where $\theta_S=\sum_{g\in S}w_g\theta_g/W_S$, $B_S=\sum_{g\in S}w_gV_g/W_S$, and
$\mathrm{Var}_S=\sum_{g\in S}w_g^2\sigma_g^2/W_S^2$, with the convention $S=\mathcal G$ when the
screen retains no cohort. When every cohort is retained, $\pi_g=1$, the two estimators
coincide and $\mathcal{A}=\mathrm{MSE}_{\mathrm{ATT}}-\mathrm{MSE}_{\mathrm{LATT}}=0$. Suppose instead
the screen separates the cohorts, discarding every confounded cohort, $\pi_g\to0$ for
$V_g\neq0$, and retaining every clean cohort, $\pi_g\to1$ for $V_g=0$, so that the retained
set converges in probability to the clean set $\mathcal C$ and the dropped set is the
confounded set $\mathcal F=\mathcal G\setminus\mathcal C$. Then
\begin{equation}\label{eq:deltastar}
\mathcal{A}\;\longrightarrow\;\mathcal{A}^\ast\;=\;B_{\mathcal F}^2-\Gamma^2-\Delta\mathrm{Var},
\end{equation}
where $B_{\mathcal F}=\sum_{g\in\mathcal F}w_gV_g$ is the aggregate violation removed by
selection (a realized quantity, unlike the worst-case level bound $B_{\hat S}$ of
Corollary~\ref{cor:causal} that widens the honest interval), $\Gamma=\theta_{\mathcal C}-\theta_{\mathcal G}$
is the composition gap between the
retained subpopulation and the full population, and
$\Delta\mathrm{Var}=\mathrm{Var}_{\mathcal C}-\sum_g w_g^2\sigma_g^2\ge0$ is the variance cost of
dropping. Hence in the informative limit the LATT dominates the ATT in mean squared error if
and only if
\begin{equation}\label{eq:generalK-threshold}
B_{\mathcal F}^2\;>\;\Gamma^2+\Delta\mathrm{Var}.
\end{equation}
With two equally weighted cohorts, one clean and one confounded with violation $V$, effect gap
$\gamma$, and common variance $\sigma^2$, one has $B_{\mathcal F}^2=V^2/4$, $\Gamma^2=\gamma^2/4$,
and $\Delta\mathrm{Var}=\sigma^2/2$, so \eqref{eq:generalK-threshold} becomes $V^2>\gamma^2+2\sigma^2$.
Remark~\ref{rem:inform} below develops this two-cohort case, where the informativeness of the
pre-trends enters.
\end{proposition}

\begin{proof}
See Appendix~\ref{app:proofs}.
\end{proof}

The general threshold \eqref{eq:generalK-threshold} is the two-cohort comparison written for many
cohorts, each of its three forces now an aggregate. The reason to trade is the total violation
$B_{\mathcal F}$ selection removes, not any single cohort's. The estimand cost is the composition gap
$\Gamma$ between the retained subpopulation's effect and the full population's, the many-cohort form of
$\gamma$. It vanishes under the equal-targets condition, that the retained and dropped cohorts share the
same average effect. The variance
cost $\Delta\mathrm{Var}$ is the extra sampling variance from reweighting onto fewer cohorts. The sign of
the advantage is a violation-against-noise-and-heterogeneity question, settled by
\eqref{eq:generalK-threshold} at the endpoints where the screen retains everyone ($\mathcal{A}=0$) or fully
separates the cohorts ($\mathcal{A}=\mathcal{A}^\ast$). Between them the path need not be monotone. With more
than two cohorts, partial selection can retain a subset of the confounded cohorts, and when their violations
cancel in the full ATT this can raise the LATT's bias, sending $\mathcal{A}$ negative before it recovers to
$\mathcal{A}^\ast$ at separation, so the monotone reading of $\mathcal{A}$ running from zero to
$\mathcal{A}^\ast$ as informativeness grows is specific to the two-type clean-versus-confounded case. Two
qualifications connect this to the inference in the main text. The separation that makes
\eqref{eq:deltastar} exact is the mean-squared-error counterpart of the condition behind
Remark~\ref{rem:oracle}. Near-threshold cohorts are retained with interior probability, the same
local-to-threshold configurations that obstruct uniform post-selection inference. And an undetectable
confounded cohort, flat before treatment yet violating parallel trends after, stays in both estimators, so
it enters \eqref{eq:deltastar} through the residual of Proposition~\ref{prop:pretest} that the honest
inference of Section~\ref{sec:honest} bounds.

The two-cohort case makes the three forces explicit and, uniquely, traces the advantage along the
informativeness axis. Consider two equally weighted cohorts sharing a common sampling variance
$\sigma^2$, observed as one pre- and one post-treatment coefficient each, with independent pre- and
post-period sampling errors. One cohort is clean, $\delta_1=0$. The other is confounded, with a
post-treatment violation $V>0$ and a pre-treatment violation $\phi V$, where $\phi\ge0$ is the
informativeness of the pre-trend. Their treatment effects are $\tau_1$ and $\tau_2$, with
$\gamma:=\tau_2-\tau_1$, so the target is the average $\bar\tau=(\tau_1+\tau_2)/2$. Taking $c$ large
enough that the clean cohort is retained with probability approaching one, the confounded cohort
survives the screen $\lvert\hat\beta_{2,\mathrm{pre}}\rvert\le c$ with probability
\begin{equation}\label{eq:p2}
p_2(\phi)\;=\;\Phi\!\Big(\tfrac{c-\phi V}{\sigma}\Big)-\Phi\!\Big(\tfrac{-c-\phi V}{\sigma}\Big),
\end{equation}
which decreases from near one, when the confound is invisible in the pre-period, to zero as the
pre-trend announces it.

\begin{remark}[Informativeness scaling, two cohorts]\label{rem:inform}
In the two-cohort model, with the clean cohort retained with probability approaching one, the
mean-squared-error advantage of the LATT estimator over the ATT estimator is
\begin{equation}\label{eq:delta}
\mathcal{A}(\phi)\;=\;\mathrm{MSE}_{\mathrm{ATT}}-\mathrm{MSE}_{\mathrm{LATT}}
\;=\;\frac{1-p_2(\phi)}{4}\,\big[\,V^2-\gamma^2-2\sigma^2\,\big],
\end{equation}
the general advantage \eqref{eq:deltastar} specialized by $B_{\mathcal F}^2=V^2/4$, $\Gamma^2=\gamma^2/4$,
$\Delta\mathrm{Var}=\sigma^2/2$ and scaled by the survival factor $1-p_2(\phi)$ (derivation in
Appendix~\ref{app:proofs}). The bracket's three terms are the forces of \eqref{eq:generalK-threshold} in
closed form, $+V^2$ the bias the screen removes, $-2\sigma^2$ the variance that dropping a cohort costs,
and $-\gamma^2$ the estimand-change cost, charged only to a reader who insists on the ATT as the common
target and vanishing under homogeneity $\gamma=0$. Hence the LATT dominates the ATT in mean squared error
if and only if
\begin{equation}\label{eq:threshold}
V^2\;>\;\gamma^2+2\sigma^2 ,
\end{equation}
a sign fixed independently of the informativeness $\phi$, which enters only through the nondecreasing
factor $1-p_2(\phi)$. Two readings follow. Informative pre-trends are necessary but not sufficient. If
$V^2<\gamma^2+2\sigma^2$ no degree of informativeness rescues the LATT, since $1-p_2$ only drives an
already-negative bracket further from zero, informativeness scales the advantage, it does not create one.
And when the bracket is positive the advantage rises monotonically in informativeness, from zero at
$\phi=0$ to its ceiling $\tfrac14[V^2-\gamma^2-2\sigma^2]$ as $\phi\to\infty$, the vertical gap between the
two bias curves in Figure~\ref{fig:scope}.
\end{remark}

The dominance condition can be read off the data, up to the one quantity it cannot identify, and so
furnishes the diagnostic reported in the switching rule (Proposition~\ref{prop:switch}). The variance cost
$\Delta\mathrm{Var}$ is a
function of $\Sigma$ and the weights alone and is therefore estimable, and the difference between the two
point estimators consistently estimates the violation net of the composition gap,
\begin{equation}\label{eq:D-identified}
\hat\theta_{\mathcal G}-\hat\theta_{\hat S}\;\xrightarrow{p}\;B_{\mathcal F}-\Gamma\;=:\;D,
\end{equation}
since the ATT estimator carries the dropped cohorts' violation while the LATT estimator does not. What is
not identified is $\Gamma$ on its own, because it depends on the counterfactual effects of the very
cohorts the screen discards. Writing $B_{\mathcal F}=D+\Gamma$ in \eqref{eq:generalK-threshold}, the
advantage $\mathcal{A}^\ast=B_{\mathcal F}^2-\Gamma^2-\Delta\mathrm{Var}$ becomes
$D^2+2D\Gamma-\Delta\mathrm{Var}$, linear in the unidentified $\Gamma$ because the squared composition gap
cancels. Two readings follow. Under homogeneous effects, $\Gamma=0$, the trade pays exactly when the
squared gap between the estimators exceeds the variance cost,
\begin{equation}\label{eq:diagnostic}
D^2\;>\;\Delta\mathrm{Var},
\end{equation}
a comparison the practitioner can compute and, using the carved standard error of
Theorem~\ref{thm:carve} for $D$, test. In general the verdict is monotone in $\Gamma$, increasing
when $D>0$, and flips at the breakdown value $\Gamma^\ast=(\Delta\mathrm{Var}-D^2)/(2D)$, the quantity
Proposition~\ref{prop:switch}(iii) reports. The trade pays unless the dropped
cohorts' effects differ from the retained ones by more than $\Gamma^\ast$ allows, a sensitivity statement
in the same idiom as the breakdown value $M^\ast$ for the identifying restriction, now for the estimand
change rather than the violation.

Three qualifications sharpen the dominance calculation. The advantage \eqref{eq:delta} is an upper bound
on what the method delivers, because it omits two channels that both work against the
LATT, occasional false rejection of the clean cohort, which enters at order
$1-p_1$ when the threshold is not comfortably large, and the additional variance that
data-driven selection induces and that a standard error conditional on the realized
set does not capture, the understatement documented in Section~\ref{sec:calib}. Both
shrink the realized advantage toward zero, and near the threshold, where the dominance margin
$V^2-\gamma^2-2\sigma^2$ is small, a positive omitted cost can reverse the ranking, so the sign is
guaranteed only when that margin exceeds these costs. Finally,
\eqref{eq:delta} is a comparison of point estimators under squared-error loss, not of
the honest intervals that are the paper's recommended output. The interval analogue
replaces the reason-to-trade $V$ with the residual violation net of the calibration
buffer of Remark~\ref{rem:calib}, and the variance cost $\sigma^2$ with a sampling
slack inflated by the selection quantile, so that the level-bound half-width comparison
turns on the same three forces with $V$ discounted and $\sigma$ inflated. The
mean-squared-error threshold \eqref{eq:threshold} correctly signs that comparison and
is the transparent summary. The width comparison of Section~\ref{sec:dominance} sharpens the
constants.

\section{Omitted proofs}\label{app:proofs}

\begin{proof}[Proof of Proposition~\ref{prop:id}]
For $g\in S$, $\beta_{g,\mathrm{post}} = \tau_{g,\mathrm{post}}+\delta_{g,\mathrm{post}}
= \tau_{g,\mathrm{post}}$ by the hypothesis and \eqref{eq:decomp}, so $a(\beta_{g,\mathrm{post}})
= a(\tau_{g,\mathrm{post}}) = \mathrm{ATT}(g,\cdot)$ by linearity of $a$. The reweighted sum over $S$
therefore equals the weighted average of the aggregated causal effects, which is $\theta_S$ by definition.
Cohorts outside $S$ receive zero weight, so their differential trends do not enter.
\end{proof}

\begin{proof}[Proof of Corollary~\ref{cor:causal}]
Theorem~\ref{thm:carve} gives $\theta_{\hat S}\in\mathcal C_{1-\alpha}(\hat S)$ with probability $1-\alpha$.
The maintained restriction places $\theta_{\hat S}^{\mathrm c}$ within $B_{\hat S}$ of $\theta_{\hat S}$
deterministically, so $\{\theta_{\hat S}\in\mathcal C_{1-\alpha}(\hat S)\}\subseteq\{\theta_{\hat S}^{\mathrm c}\in
\mathcal C^{M}_{1-\alpha}(\hat S)\}$, and monotonicity of probability gives the bound.
\end{proof}

\begin{proof}[Proof of Proposition~\ref{prop:nolength}]
By Theorem~\ref{thm:carve}, conditional on $\{\hat S=S\}$ and the orthogonal remainder $r$,
$\hat\theta_S\sim\mathcal{TN}(\theta_S,s_S^2,[\mathcal V^-,\mathcal V^+])$, and $\mathcal C_{1-\alpha}(S)$
inverts the truncated-normal pivot in $\theta_S$, which is exactly the confidence set of \citet{lee2016}
for a normal mean under a polyhedral, hence after conditioning on $r$ interval, truncation.
\citet{kivaranovicleeb2021} show this set has infinite expected length whenever the truncation is one-sided
with positive probability. Under the flatness screen a retained cohort's binding constraint is two-sided
only when both faces $\lvert\hat\beta_{g,\mathrm{pre}}(e)+\omega_g\rvert\le c$ are active; generically one
face binds and the induced bound on $\hat\theta_S$ is one-sided, an event of positive probability, so the
hypothesis holds. Averaging over $S$ preserves the conclusion. That only $\gamma\to\infty$ escapes is
immediate, since it empties the selection.
\end{proof}

\begin{proof}[Proof of Proposition~\ref{prop:switch}]
(i) Each interval's coverage is a property of that interval alone and holds for every $\beta$. Reporting
both leaves each untouched, so whichever the rule recommends, the reader is handed an interval that covers
at $1-\alpha$. (ii) By Theorem~\ref{thm:width}, $h_S<h_{\mathcal G}\iff B_{\mathcal G}>B_{\mathcal G}^\ast$,
and both half-widths are functions of $\hat\Sigma$ and the chosen bounds. For a fixed retained set
$\hat D\xrightarrow{p}-\Gamma+(\delta_{\mathcal G}-\delta_S)$, which equals $B_{\mathcal F}-\Gamma$ in the
informative limit where the retained aggregate violation vanishes and the full aggregate violation is
carried by the dropped cohorts. Its standard error is that of the contrast
$(\ell_{\mathcal G}-\ell_S)'\hat\beta_{\mathrm{post}}$, carved for the selection. (iii) The composition gap
$\Gamma=\theta_{\hat S}-\theta_{\mathcal G}$ depends only on
the discarded cohorts' effects and is therefore unidentified. The explicit $\Gamma^\ast$ follows from the
mean-squared-error advantage of Appendix~\ref{app:mse}.
\end{proof}

\begin{proof}[Proof of Proposition~\ref{prop:generalK}]
The ATT estimator is a fixed linear combination with mean
$\theta_{\mathcal G}+\sum_g w_gV_g$ and variance $\sum_g w_g^2\sigma_g^2$, giving
$\mathrm{MSE}_{\mathrm{ATT}}$. For the LATT, condition on the retention configuration. On the
event that the retained set is $S$, which has probability
$\prod_{g\in S}\pi_g\prod_{g\notin S}(1-\pi_g)$, the estimator is
$\sum_{g\in S}w_g\hat\beta_{g,\mathrm{post}}/W_S$, whose mean is $\theta_S+B_S$ and whose
variance is $\mathrm{Var}_S$, since retention is independent of the post-treatment estimates. Its
error relative to $\theta_{\mathcal G}$ has squared bias $(\theta_S-\theta_{\mathcal G}+B_S)^2$ and
variance $\mathrm{Var}_S$. Averaging over configurations gives \eqref{eq:latt-exact}. At
$\pi_g=1$ only $S=\mathcal G$ has mass, where $\theta_{\mathcal G}-\theta_{\mathcal G}+B_{\mathcal G}
=\sum_g w_gV_g$ and $\mathrm{Var}_{\mathcal G}=\sum_g w_g^2\sigma_g^2$, so
$\mathrm{MSE}_{\mathrm{LATT}}=\mathrm{MSE}_{\mathrm{ATT}}$ and $\mathcal{A}=0$. Under separation the
configuration $S=\mathcal C$ carries probability approaching one, so
$\mathrm{MSE}_{\mathrm{LATT}}\to(\theta_{\mathcal C}-\theta_{\mathcal G})^2+\mathrm{Var}_{\mathcal C}$,
using $B_{\mathcal C}=0$ because clean cohorts have $V_g=0$. Since clean cohorts contribute nothing
to $\sum_g w_gV_g$, that sum equals $B_{\mathcal F}$, and
$\mathcal{A}\to B_{\mathcal F}^2+\sum_g w_g^2\sigma_g^2-\Gamma^2-\mathrm{Var}_{\mathcal C}
=B_{\mathcal F}^2-\Gamma^2-\Delta\mathrm{Var}$, which is \eqref{eq:deltastar}. The threshold
\eqref{eq:generalK-threshold} is its positivity, and the two-cohort reduction is the displayed
substitution.
\end{proof}

\begin{proof}[Proof of Remark~\ref{rem:inform}]
Write $\tau_1=\bar\tau-\gamma/2$, $\tau_2=\bar\tau+\gamma/2$, so
$\hat\beta_{1,\mathrm{post}}\sim\mathcal N(\bar\tau-\gamma/2,\sigma^2)$ and
$\hat\beta_{2,\mathrm{post}}\sim\mathcal N(\bar\tau+\gamma/2+V,\sigma^2)$, independent of each other
and of the retention indicators. The ATT estimator has
$\mathrm{MSE}_{\mathrm{ATT}}=V^2/4+\sigma^2/2$. With the clean cohort always retained, the LATT
estimator is the two-cohort average with probability $p_2$ and $\hat\beta_{1,\mathrm{post}}$ with
probability $1-p_2$, so
$\mathrm{MSE}_{\mathrm{LATT}}=p_2(V^2/4+\sigma^2/2)+(1-p_2)(\gamma^2/4+\sigma^2)$. Subtracting gives
\eqref{eq:delta}, which is \eqref{eq:deltastar} under the stated substitution. Since $1-p_2\ge0$ the
sign is that of the bracket, giving \eqref{eq:threshold}, and $p_2(\phi)$ is nonincreasing in $\phi$
for $\phi V\ge0$, so $1-p_2$ is nondecreasing.
\end{proof}

\begin{proof}[Proof of Theorem~\ref{thm:width}]
By \eqref{eq:flci-hw} each interval is the near-optimal fixed-length interval of
\citet{armstrong2018}, with half-width $h(B,s)=s\,\mathrm{cv}_\alpha(B/s)$ and coverage at least
$1-\alpha$ whenever the bound dominates the estimator's bias. For the ATT this gives
$h_{\mathcal G}=h(B_{\mathcal G},s_{\mathcal G})$, and under ex-ante $S$ the same applies to the
LATT, so $h_S=h(B_S,s_S)$. Both half-widths then share $\varphi(B,s)=s\,\mathrm{cv}_\alpha(B/s)$, so
$h_S<h_{\mathcal G}$ is the displayed inequality. Under data-driven selection
$\hat\theta_S\mid\{\hat S=S\}$ is truncated Gaussian, so the honest LATT interval is the carved interval of
Theorem~\ref{thm:carve} rather than the fixed-length \eqref{eq:flci-hw}. Its length is data-dependent and
generally asymmetric, and its conditional expected length is infinite by
Proposition~\ref{prop:nolength}, so no expected-length ordering is well-posed; the displayed threshold is
therefore stated for the fixed-length intervals of ex-ante selection only.
Fix $(s_{\mathcal G},s_S,B_S)$ and set
$c_0=s_S\,\mathrm{cv}_\alpha(B_S/s_S)$. The map $g(B)=s_{\mathcal G}\,\mathrm{cv}_\alpha(B/s_{\mathcal G})$
is continuous and strictly increasing, from $g(0)=s_{\mathcal G}\,z_{1-\alpha/2}$ to $\infty$. Since
$\mathrm{cv}_\alpha:[0,\infty)\to[z_{1-\alpha/2},\infty)$ is a strictly increasing bijection and
when $s_S\ge s_{\mathcal G}$ we have $c_0/s_{\mathcal G}\ge (s_S/s_{\mathcal G})\,z_{1-\alpha/2}\ge
z_{1-\alpha/2}$, so the equation $g(B_{\mathcal G}^\ast)=c_0$ has the unique solution displayed and
$h_S<h_{\mathcal G}\iff B_{\mathcal G}>B_{\mathcal G}^\ast$. When $s_S<s_{\mathcal G}$ the argument
$c_0/s_{\mathcal G}$ can fall below $z_{1-\alpha/2}$, $g(B_{\mathcal G}^\ast)=c_0$ has no nonnegative
solution, and $h_S<h_{\mathcal G}$ already at $B_{\mathcal G}=0$. Under precision weights, reweighting onto
the retained set cannot lower the aggregate variance below the full set's optimum, so $s_S\ge s_{\mathcal G}$
and $B_{\mathcal G}^\ast>0$ whenever $s_S>s_{\mathcal G}$. Under other weightings the ordering must be checked.
\end{proof}

\end{document}